\documentclass[a4paper,envcountsame,10pt]{article}
\usepackage{authblk}

\usepackage{ifthen}

\newboolean{talk}
\setboolean{talk}{false}
\newboolean{paper}
\setboolean{paper}{false}

\newboolean{LMCSstyle}
\setboolean{LMCSstyle}{false}
\newboolean{IEEEstyle}
\setboolean{IEEEstyle}{false}
\newboolean{lipicsstyle}
\setboolean{lipicsstyle}{false}
\newboolean{eptcsstyle}
\setboolean{eptcsstyle}{false}
\newboolean{sigplanstyle}
\setboolean{sigplanstyle}{false}
\newboolean{PACMPL}
\setboolean{PACMPL}{false}
\newboolean{acmartstyle}
\setboolean{acmartstyle}{false}

\newboolean{needstheorems}
\setboolean{needstheorems}{false}
\newboolean{withimages}
\setboolean{withimages}{false}
\newboolean{withproofs}
\setboolean{withproofs}{true}

\newboolean{french}
\setboolean{french}{false}

\setboolean{needstheorems}{true}
\ifthenelse{\boolean{talk}}{}{\usepackage[utf8]{inputenc}}
\ifthenelse{\boolean{PACMPL}}{}{
	\ifthenelse{\boolean{french}}{\usepackage[french]{babel}}{
	  \ifthenelse{\boolean{lipicsstyle}}{}{\usepackage[english]{babel}}}
	  }
\ifthenelse{\boolean{PACMPL}}{}{
	\ifthenelse{\boolean{acmartstyle}}{}{
		\usepackage{amssymb} 
	}
}
\usepackage{amsmath}
\usepackage{graphicx}
\usepackage{bussproofs}
\usepackage{fancybox}
\usepackage{cmll}
\usepackage{stmaryrd}
\usepackage{multirow}
\ifthenelse{\boolean{IEEEstyle}}{
\usepackage[bookmarks={false}]{hyperref}}{
\usepackage{hyperref}}
\usepackage{proof}
\usepackage{hhline}
\usepackage{xspace}
\usepackage{booktabs}
\usepackage{subcaption}
\usepackage{wrapfig}
\usepackage{array}
\usepackage{tabularx}
\usepackage{arydshln}
\usepackage{commath}
\ifthenelse{\boolean{IEEEstyle}}{
\setboolean{needstheorems}{true}}{}

\ifthenelse{\boolean{eptcsstyle}}{
\setboolean{needstheorems}{true}}{}

\ifthenelse{\boolean{sigplanstyle}}{
\setboolean{needstheorems}{true}}{}

\ifthenelse{\boolean{needstheorems}}{

\usepackage{amsthm}

    \newtheorem{theorem}{Theorem}[section]
    \newtheorem{lemma}[theorem]{Lemma}
    \newtheorem{corollary}[theorem]{Corollary}

}{}

\newcommand{\myproof}[1]{
\ifthenelse{\boolean{withproofs}}{#1}{}}

\newcommand{\withproofs}[1]{
\ifthenelse{\boolean{withproofs}}{#1}{}}

\newcommand{\withoutproofs}[1]{
\ifthenelse{\boolean{withproofs}}{}{#1}}

\newcommand{\tm}{t}
\newcommand{\tmtwo}{u}
\newcommand{\tmthree}{r}
\newcommand{\tmfour}{w}
\newcommand{\tmfive}{s}

\newcommand{\var}{x}
\newcommand{\vartwo}{y}
\newcommand{\varthree}{z}

\newcommand{\rootRew}[1]{\mapsto_{#1}}
\newcommand{\Rew}[1]{\rightarrow_{#1}}

\newcommand{\rtob}{\rootRew{\beta}}
\newcommand{\tob}{\Rew{\beta}}

\newcommand{\symfont}[1]{\mathsf{#1}}

\newcommand{\val}{v}
\newcommand{\valtwo}{w}
\newcommand{\valthree}{v'}

\newcommand{\ctxholep}[1]{[#1]}
\newcommand{\ctxhole}{\ctxholep{\cdot}}

\newcommand{\evctx}{E}

\newcommand{\evctxp}[1]{\evctx\ctxholep{#1}}

\newcommand{\nbvctxtwo}[1]{\nbvctxtwo{#1}}

\newcommand{\defeq}{:=}

\newcommand{\grameq}{::=}

\newcommand{\isub}[2]{\{#1/#2\}}

\newcommand{\llbrace}{\{ \kern -0.27em \vert}
\newcommand{\rrbrace}{\vert \kern -0.27em \}}

\renewcommand{\l}{\lambda}
\newcommand{\ie}{i.e.\xspace}
\ifthenelse{\boolean{LMCSstyle}}{}{
	
	}
\newcommand{\ih}{\textit{i.h.}\xspace}

\newcommand{\ignore}[1]{}

\newcommand{\myinput}[1]{\ifthenelse{\boolean{withimages}}{\input{#1}}{}}

\newcommand{\nat}{\mathbb{N}}

\newcommand{\size}[1]{|#1|}

\ifthenelse{\boolean{PACMPL}}{\renewcommand{\state}{s}}{
	\ifthenelse{\boolean{acmartstyle}}{\renewcommand{\state}{s}}{
		\newcommand{\state}{s}
	}
}

\newcounter{numberone}

\newcounter{numbertwo}

\renewcommand{\ctxholep}[1]{\langle #1\rangle}

\newcommand{\tmsix}{q}

\newcommand{\reflemma}[1]{Lemma~\ref{l:#1}}

\newcommand{\reflemmaeq}[1]{L.~\ref{l:#1}}

\ifthenelse{\boolean{LMCSstyle}}{	
	
	}{
	
	}

\renewcommand{\isub}[2]{\{#1{\shortleftarrow}#2\}}

\newcommand{\run}{\rho}

\newcommand{\cons}{{\cdot}}

\newcommand{\tomachhole}[1]{\rightarrow_{#1}}

\newcommand{\la}[1]{\lambda #1.}

\newcommand{\bigo}[1]{\mathcal{O}(#1)}

\newcommand{\Id}{\symfont{I}}

\newcommand\mydots{\hbox to .6em{.\hss.}}

\newcommand{\str}{s}

\newcommand{\enc}[1]{\lceil#1\rceil}
\newcommand{\ems}{\varepsilon}

\newcommand{\detLam}{\Lambda_{\tt det}}
\newcommand{\tobdet}{\rightarrow_{det}}

\newcommand{\alpone}{\Sigma}
\newcommand{\cods}[1]{\overline{#1}}

\newcommand{\refcoroeq}[1]{Cor.~\ref{coro:#1}}

\newcommand{\fixnospace}{\symfont{fix}}
\newcommand{\fix}{\fixnospace\,}
\newcommand{\M}{\mathcal M}
\newcommand{\cont}{k}
\newcommand{\init}[2]{{\tt init}^{#1}_{#2}}
\newcommand{\inits}{{\tt init}}

\newcommand{\elemblank}{\Box}
\renewcommand{\state}{q}
\newcommand{\States}{Q}
\newcommand{\statein}{\state_{\mathit{in}}}
\newcommand{\statefint}{\state_{\mathit{T}}}
\newcommand{\statefinf}{\state_{\mathit{F}}}

\newcommand{\config}{C}
\newcommand{\configtwo}{D}
\newcommand{\final}[2]{{\tt final}^{#1}_{#2}}
\newcommand{\finals}{{\tt final}}

\newcommand{\transaux}{{\tt transaux}}
\newcommand{\trans}[1]{{\tt trans}^{#1}}
\newcommand{\transs}{{\tt trans}}

\newcommand{\strone}{s}
\newcommand{\strtwo}{r}
\newcommand{\strthree}{p}
\newcommand{\append}{{\tt{append}}}
\newcommand{\appendalpchar}[2]{\append_{#1}^{#2}}

\newcommand{\appendchar}[1]{\append^{#1}}

\newcommand{\elone}{a}
\newcommand{\elem}{a}
\renewcommand{\enc}[1]{\overline{#1}}
\renewcommand{\cods}[1]{\lceil#1\rceil}
\newcommand{\cod}[2]{\cods{#1}^{#2}}

\newcommand{\initconfig}{\config_{\tt in}^\M(\strone)}
\newcommand{\initconfigs}{\config_{\tt in}(\inputstr)}

\newcommand{\tomachtur}{\tomachhole{\M}}
\newcommand{\Nset}{\mathbb{N}}
\newcommand{\inputstr}{i}
\newcommand{\counter}{n}
\newcommand{\Statesfin}{\States_{\mathit{fin}}}

\newcommand{\tuple}[1]{\langle #1 \rangle}

\newcommand{\succf}{\textsc{succ}}
\newcommand{\succl}{\mathtt{succ}}
\newcommand{\succaux}{\texttt{succaux}}

\newcommand{\predf}{\textsc{pred}}
\newcommand{\predl}{\mathtt{pred}}
\newcommand{\predaux}{\texttt{predaux}}

\newcommand{\lookupf}{\textsc{lookup}}
\newcommand{\lookupl}{\mathtt{lookup}}
\newcommand{\lookupaux}{\mathtt{lookupaux}}

\newcommand{\ntostr}[1]{\lfloor #1 \rfloor}

\newcommand{\rop}{\textsc{r}}

\newcommand{\Bool}{\mathbb{B}}
\newcommand{\Boolb}{\Bool_{\mathsf{W}}}
\newcommand{\Boolin}{\Bool_{\mathsf{I}}}
\newcommand{\csep}{\,|\,}
\newcommand{\encp}[2]{\enc{#1}^{#2}}
\newcommand{\bool}{b}

\newcommand{\wstr}{w}
\newcommand{\wstrl}{\wstr_{l}}
\newcommand{\wstrr}{\wstr_{r}}

\newcommand{\istr}{i}

\newcommand{\varf}{f}

\newcommand{\stsym}{\mathsf{L}}
\newcommand{\ensym}{\mathsf{R}}

\usepackage{microtype}
\usepackage{version}

\usepackage{microtype}
\usepackage{version}

\includeversion{LONG}
\excludeversion{SHORT}

\begin{document}
%
\title{A Log-Sensitive Encoding of\\ Turing Machines in the $\l$-Calculus}

\author[1]{Beniamino Accattoli}
\affil[1]{Inria, \& LIX, \'Ecole Polytechnique, UMR 7161, France}
\author[2]{Ugo Dal Lago}
\author[2]{Gabriele Vanoni}
\affil[2]{Universit\'a di Bologna \& Inria Sophia Antipolis, Italy}

\date{}


\maketitle

\begin{abstract}
This note modifies the reference encoding of Turing machines in the $\l$-calculus by Dal Lago and Accattoli \cite{DBLP:journals/corr/abs-1711-10078}, which is tuned for time efficiency, as to accommodate logarithmic space. There are two main changes: Turing machines now have \emph{two} tapes, an input tape and a work tape, and the input tape is encoded differently, because the reference encoding comes with a linear space overhead for managing tapes, which is excessive for studying logarithmic space.
\end{abstract}

\section{Introduction}
This note presents a new encoding of Turing machines into the $\l$-calculus and and proves its correctness. It is based over Dal Lago and Accattoli's reference encoding of single tape Turing machines
\cite{DBLP:journals/corr/abs-1711-10078}. The new encoding is tuned for studying logarithmic space complexity even though such a study is not carried out here but in a companion paper.  The aim of this note is to provide the formal definition of the encoding and the tedious calculations to prove its correctness. 

The key points of the new encoding with respect to the reference one are:
\begin{itemize}
\item \emph{Log-sensitivity}: the reference encoding cannot account for logarithmic space complexity because it is based on Turing machines with a single tape. Logarithmic space requires Turing machines with a read-only input tape, of input string $\istr$, and an ordinary work tape, and to only count the space used on the work tape---if one counts the space for $\istr$, it is impossible to use only $\bigo{\log\size\istr}$ space. We refer to such a separation of tapes as to \emph{log-sensitivity}. The reference encoding is log-\emph{in}sensitive instead.
\item \emph{Different encoding of the input tape}: simply adapting the reference encoding to two tapes is not enough for  preserving logarithmic space, because the reference encoding is tuned for time: it reads from tapes in $\bigo{1}$ time  but the reading mechanism comes with a $\bigo{\size\istr}$ space overhead, while the input tape has to be handled with at most $\bigo{\log\size\istr}$ space overhead. We then change the encoding and the reading mechanism of the input tape, trading time for space, as to read in $\bigo{\size\istr\log\size\istr}$ time with $\bigo{\log\size\istr}$ space overhead. The idea is that the position of the head is indicated by a pointer, given by the position index (of logarithmic size), and reading from the input requires scrolling the input tape sequentially, until the position index is reached.
\item \emph{Different time complexity}: by trading time for space, the new encoding is slower. If a Turing machine $\M$ takes time $T_{\M}(\size\istr)$ on input string $\istr$ then the encoding evaluates in  $\Theta((T_{\M}(\size\istr)+1)\cdot \size\istr\cdot 
	\log{\size\istr})$ $\beta$-steps rather than in $\Theta(T_{\M}(\size\istr))$ as in the reference encoding, because each transition reads from the input. The time complexity is however still polynomial and thus the encoding is reasonable for time (considered as the number of $\beta$-steps).
\end{itemize} 

\paragraph{Intrinsic and Mathematical Tape Representations}  A TM tape is a string plus a distinguished position, representing the head. There are two tape representations, dubbed \emph{intrinsic} and \emph{mathematical} by van Emde Boas in  \cite{DBLP:conf/sofsem/Boas12}, described next.
\begin{itemize}
\item The \emph{intrinsic} one represents both the string $\istr$ and the current position of the head as the triple $\istr = \istr_{l} \cdot h \cdot \istr_{r}$, where $\istr_{l}$ and $\istr_{r}$ are the prefix and suffix of $\istr$ surrounding the character $h$ read by the head. 

The reference encoding of TMs in the $\l$-calculus uses the intrinsic representation of tapes and, as already mentioned, reading costs $\bigo{1}$ time but the reading mechanism comes with a $\bigo{\size\istr}$ space overhead.

\item
The \emph{mathematical} representation of tapes, instead, is simply given by the index $n\in\nat$ of the head position, that is, the triple $\istr_{l} \cdot h \cdot \istr_{r}$ is replaced by the pair $(\istr, \size{\istr_{l}}+1)$. The index $\size{\istr_{l}}+1$ has the role of a pointer, of logarithmic size when represented in binary. 

An encoding in the $\l$-calculus based on the mathematical representation reads in $\bigo{\size\istr\log\size\istr}$ time with $\bigo{\log\size\istr}$ space overhead. The time cost is due to the fact that accessing the tape is done sequentially, by moving of one cell at a time on the tape and at each step decrementing of one the index, until the index is $0$. Therefore, accessing the right cell requires to decrement the index $\size{\istr_{l}}+1$ times, each time requiring time $\log{\size{\istr_{l}}+1}$, because the index is represented in binary.
\end{itemize}
The new log-sensitive encoding keeps the intrinsic representation for the work tape, and instead adopts the mathematical representation for the input tape. This is because linear space overhead for the work tape is not a problem for log-sensitivity, while the mathematical representation makes it harder to write on the tape.

\paragraph{Binary Arithmetic and Literature.} For implementing on the $\l$-calculus the manipulation of the head index, one needs to develop an encoding of binary strings and three operations: successor, predecessor, and lookup of the $n$-th character of a string starting from the binary representation of $n$. We do this using the Scott encoding of binary string, but using a reversed representation of binary strings. 

We developed our encoding of binary arithmetic from scratch, in a seemingly ad-hoc way. We discovered afterwards that our approach is a variation over encodings in the literature. Namely, it is quite similar to Mogensen's encoding \cite{DBLP:conf/ershov/Mogensen01}, itself building over Goldberg's study \cite{DBLP:journals/jfp/Goldberg00}.


\section{The Deterministic $\l$-Calculus and the Scott Encoding of Strings}
As it is the case for the reference encoding, also the new encodings has as image a restricted form of $\l$-calculus, what we refer to as the deterministic $\l$-calculus, where, in particular, call-by-name and call-by-value evaluation coincide.

\paragraph{Deterministic $\l$-Calculus.}
The language and the evaluation contexts of the \emph{deterministic $\l$-calculus $\detLam$} are given by:
\begin{center}
	$\begin{array}{r@{\hspace{.5cm}}rlllllll}
	\mbox{Terms} & \tm,\tmtwo,\tmthree,\tmfour &\grameq& \val \mid  \tm \val\\
	\mbox{Values} & \val, \valtwo, \valthree & \grameq & \la\var\tm \mid \var\\\\
	\mbox{Evaluation Contexts} & \evctx  & \grameq & \ctxhole\mid  \evctx\val
	\end{array}$
\end{center}

Note that 
\begin{itemize}
	\item \emph{Arguments are values}: the right subterm of an application has to be a value, in contrast to what happens in the ordinary $\l$-calculus. 
	\item \emph{Weak evaluation}: evaluation contexts are \emph{weak}, \ie they do not enter inside abstractions. 
\end{itemize}

Evaluation is then defined by:
\begin{center}
	$\begin{array}{c@{\hspace{.5cm}}rll}
	\textsc{Rule at top level} & \multicolumn{3}{c}{\textsc{Contextual closure}} \\
	\!(\la\var\tm)\tmtwo \rtob \tm \isub\var\tmtwo &
	\multicolumn{3}{c}{\evctxp \tm \tobdet \evctxp \tmtwo \textrm{~~~if } \tm \rtob \tmtwo} \\
	\end{array}$
\end{center}

\emph{Convention}: to improve readability we omit some parenthesis, giving precedence to application with respect to abstraction. Therefore $\la \var \tm \tmtwo$ stands for $\la \var (\tm \tmtwo)$ and not for $(\la \var \tm) \tmtwo$, that instead requires parenthesis.

The name of this calculus is motivated by the following immediate lemma.

\begin{lemma}[Determinism]
\label{l:lamdet-is-det}
Let $\tm \in \detLam$. If $\tm \tobdet \tmtwo$ and $\tm \tobdet \tmthree$ then $\tmtwo = \tmthree$.
\end{lemma}

\begin{proof}
By induction on $\tm$. If $\tm$ is a variable or an abstraction then it cannot reduce. If $\tm = \tmfour \val$ then there are two cases:
\begin{itemize}
\item \emph{$\tmfour$ is an abstraction $\la\var\tmfive$}. Then $\tm = (\la\var\tmfive)\val \tobdet \tmfive\isub\var\val$ is the unique redex of $\tm$, that is, $\tmtwo=\tmthree=\tmfive\isub\var\val$.
\item \emph{$\tmfour$ is not an abstraction}. Then the two steps from $\tm$ come from two steps $\tmfour \tobdet \tmtwo'$ and $\tmfour \tobdet \tmthree'$ with $\tmtwo = \tmtwo' \val$ and $\tmthree= \tmthree' \val$, because $\ctxhole\val$ is the only possible evaluation context. By \ih, $\tmtwo' = \tmthree'$, that is, $\tmtwo=\tmthree$.\qedhere
\end{itemize}
\end{proof}

\paragraph{Fixpoint.} The encoding of Turing machines requires a general form of recursion, that is usually implemented via a fixpoint combinator. We use Turing's fixpoint combinator, in its call-by-value variant, that fits into $\detLam$ and that returns a fixpoint up to $\eta$-equivalence. Let $\fix$ be the term $\theta\theta$, where
$$
\theta \defeq \la\var \la\vartwo \vartwo(\la\varthree \var \var \vartwo 
\varthree).
$$
Now, given a term $\tmtwo$ let us show that $\fix\tmtwo$ is a fixpoint of 
$\tmtwo$ up to $\eta$-equivalence.
$$\begin{array}{lcllll}
	\fix\tmtwo  
	& = &
	(\la\var \la\vartwo \vartwo(\la\varthree \var \var \vartwo \varthree)) 
	\theta 
	\tmtwo \\
	& \tobdet & 
	(\la\vartwo \vartwo(\la\varthree \theta \theta \vartwo \varthree)) \tmtwo \\
	
	& \tobdet & 
	\tmtwo (\la\varthree \theta \theta \tmtwo \varthree)\\ 
	& = &
	\tmtwo (\la\varthree \fix \tmtwo \varthree)\\
	& =_\eta &
	\tmtwo (\fix \tmtwo)
\end{array}$$
It is well-known that $\eta$-equivalent terms are indistinguishable in the $\l$-calculus (this is B\"ohm's theorem). Therefore, we will simply use the fact that $\fix\tmtwo \tobdet^2 \tmtwo (\la\varthree \fix \tmtwo \varthree)$ without dealing with $\eta$-equivalence. This fact will not induce any complication.

\paragraph{Encoding alphabets.} Let $\alpone=\{\elone_1,\ldots,\elone_n\}$ be a 
finite alphabet. Elements of $\alpone$ are encoded
as follows:
$$
\cod{\elone_i}{\alpone} \defeq \lambda\var_1.\ldots.\lambda\var_n.\var_i \;.
$$
When the alphabet will be clear from the context we will simply write 
$\cods\elone_i$. Note that 
\begin{enumerate}
	\item the representation fixes a total order
	on $\alpone$ such that $\elem_i< \elem_j$ iff $i< j$;
	\item the representation of an element $\cod{\elone_i}{\alpone}$ requires 
	space linear (and not logarithmic) in $\size\alpone$. But, since $\alpone$ 
	is fixed, it actually requires constant space.
\end{enumerate}

\paragraph{Encoding strings.} A string in
$\strone\in\alpone^*$ is represented by a term
$\encp\strone{\alpone^*}$. Our encoding exploits the fact that a string is a concatenation of characters \emph{followed by the empty string $\varepsilon$} (which is generally omitted). For that, the encoding uses $\size\alpone+1$ abstractions, the extra one ($\var_\varepsilon$ in the definition below) being used to represent $\varepsilon$. The encoding is defined by induction on the structure of $\strone$ as follows:
\begin{align*}
	\enc{\varepsilon}^{\alpone^*} & \defeq 
	\la{\var_1}\ldots\la{\var_n}\la{\var_\varepsilon}\var_\varepsilon\;,\\
	\enc{\elone_i\strtwo}^{\alpone^*} & \defeq 
	\la{\var_1}\ldots\la{\var_n}\la{\var_\varepsilon}\var_i\enc{\strtwo}^{\alpone^*}.
\end{align*}
Note that the representation depends on the cardinality
of $\alpone$. As before, however, the alphabet is a fixed parameter, and so such a 
dependency is irrelevant. As an example, the encoding of the string $aba$ with respect to the alphabet $\set{a,b}$ ordered as $a<b$ is 
\begin{center}
$\enc{aba}^{\set{a,b}} = \la{\var_a}\la{\var_b}\la{\var_\varepsilon}\var_a (\la{\var_a}\la{\var_b}\la{\var_\varepsilon}\var_b(\la{\var_a}\la{\var_b}\la{\var_\varepsilon}\var_a(\la{\var_a}\la{\var_b}\la{\var_\varepsilon}{\var_\varepsilon})))$
\end{center}

\begin{lemma}[Appending a character in constant time]
	\label{l:append-char}
	Let $\alpone$ be an alphabet and $\elem \in\alpone$ one of its characters. 
	There is 
	a term $\appendalpchar\alpone\elem$ 
	such that for every
	continuation term $\cont$ and every string $\strone\in\alpone^*$, 
	$$
	\appendalpchar\alpone \elem \cont  \enc\strone
	\tobdet^{\bigo{1}} 
	\cont \enc{(\elem\strone)}.
	$$
\end{lemma}
\begin{proof}
	Define the term $\appendalpchar \alpone \elem \defeq \la{\cont'} 
	\la{\strone'} \cont'(\lambda \var_1.\ldots.\lambda 
	\var_{\size\alpone}.\lambda {\var_\varepsilon}.\var_{i_\elem} \strone')$ where 
	$i_\elem$ is the index of $\elem$ in 
	the ordering of $\alpone$ fixed by its encoding, that appends the character 
	$\elem$ to the string $\strone'$ relatively to the alphabet $\alpone$. We 
	have:
	\begin{center}$\begin{array}{rclcl}
		\appendalpchar \alpone  \elem \cont \enc\strone
		& = &
		(\la{\cont'} \la{\strone'} \cont'(\lambda \var_1.\ldots.\lambda 
		\var_{\size\alpone}.\lambda {\var_\varepsilon}.\var_{i_\elem} \strone')) 
		\cont\enc\strone  \\
		& \tobdet^2 &
		\cont(\lambda \var_1.\ldots.\lambda \var_{\size\alpone}.\lambda 
		{\var_\varepsilon}.\var_{i_\elem}\enc{\strone}) \\
		& = &
		\cont \enc{(\elem \strone)}.
	\end{array}$
	\end{center}
	
\end{proof}

\subsection{Binary Arithmetic}

In order to navigate the input word, we consider a counter (in binary). Moving 
the head left (respectively right) amounts to decrement (respectively 
increment) the counter by one. The starting idea is to see a number as its 
binary string representation and to use the Scott encoding of strings. Since it 
is tricky to define the successor and predecessor on such an encoding, we 
actually define an ad-hoc encoding.

The first unusual aspect of our encoding is that the binary string is 
represented in reverse order, so that the representation of $2$ is $01$ and not 
$10$. This is done to ease the definition of the successor and predecessor 
functions as $\l$-terms, which have to process strings from left to right, and 
that with the standard representation would have to go to the end of the string 
and then potentially back from right to left. With a reversed representation, 
these functions need to process the string only once from left to right.

The second unusual aspect is that, in order to avoid problems with strings made 
out of all $0$s and strings having many $0$s on the right (which are not 
meaningful), we collapse all suffixes made out of all $0$ on to the empty 
string. A consequence is that the number $0$ is then represented with the empty 
string. Non-rightmost $0$ bits are instead represented with the usual Scott 
encoding.

If $n\in\nat$ we write $\ntostr n$ for the binary string representing $n$. Then 
we have:

\[
\begin{array}{lll}
	\ntostr 0 &\defeq&\ems
	\\[3pt]
	\ntostr 1 &\defeq&1
	\\[3pt]
	\ntostr 2 &\defeq& 01
	\\[3pt]
	\ntostr 3 &\defeq& 11
	\\[3pt]
	\ntostr 4 &\defeq& 001
	\\
\end{array}
\]
And so on. Binary strings are then encoded as $\l$-terms using the Scott 
encoding, as follows:
\[
\begin{array}{lll}
	\enc{\ems}&\defeq&\la{\var_0}{\la{\var_1}{\la{\var_\varepsilon}{\var_\varepsilon}}}
	\\[3pt]
	\enc{0\cons\str}&\defeq&\la{\var_0}{\la{\var_1}{\la{\var_\varepsilon}\var_0\enc{\str}}}
	\\[3pt]
	\enc{1\cons\str}&\defeq&\la{\var_0}{\la{\var_1}{\la{\var_\varepsilon}\var_1\enc{\str}}}
	\\
\end{array}
\]

\paragraph{Successor Function.} The successor function $\succf$ on the reversed 
binary representation can be defined as follows (in Haskell-like syntax):
\[
\begin{array}{llll}
	\succf &\ems& = &1\\
	\succf &0\cons\str& = &1\cons\str\\
	\succf &1\cons\str& = &0\cons(\succf\,\str)\\
\end{array}
\]

For which we have $\succf(\ntostr n) = \ntostr{n+1}$.

\begin{lemma}
	\label{l:succ}
	There is 
	a $\l$-term $\succl$ 
	such that for every
	continuation term $\cont$ and every natural number $\counter\in\nat$, 
	$$
	\succl\, \cont  \enc{\ntostr\counter}
	\tobdet^{\bigo{\log\counter}} 
	\cont\, \enc{\succf\, \ntostr\counter}.
	$$
\end{lemma}

\begin{proof}
Define 
$\succl \defeq \Theta\succaux$ and
$\succaux \defeq \la\varf{\la {\cont'}{\la {\counter'}{\counter'N_0N_1N_{\ems} \varf \cont'}}}$
where:
\begin{itemize}
	\item $N_0\defeq \la{\varf'}\la {\str'}{\la {\cont'}{\appendchar1\cont' \str'}}$
	\item $N_1\defeq \la{\varf'}\la {\str'}{\la {\cont'}{\varf'(\la\varthree\appendchar0\cont' \varthree) \str'}}$
	\item $N_{\ems}\defeq \la{\varf'}\la {\cont'} \cont'\, \enc{1\cons\ems}$
\end{itemize}
We rather prove $
	\succl\, \cont  \enc{\ntostr\counter}
	\tobdet^{\bigo{\size{\ntostr\counter}}} 
	\cont\, \enc{\succf\, \ntostr\counter}.
	$, where clearly $\size{\ntostr\counter} = \log \counter$, because the proof is naturally by induction on the length of $\ntostr\counter$ as a string. The first steps of the evaluation of $\succl\, k \enc{\ntostr\counter}$ are common to all natural numbers $\counter\in\nat$:
\[\begin{array}{llllllll}
\succl\, \cont \enc{\ntostr\counter} & = & \fix \succaux\, \cont \enc{\ntostr\counter}
\\ & \tob^{2} &
\succaux (\la\varthree \succl\, \varthree) \cont \enc{\ntostr\counter}
\\ & = &
(\la\varf{\la {\cont'}{\la {\counter'}{\counter'N_0N_1N_{\ems} \varf \cont'}}}) (\la\varthree \succl\,\varthree) \cont \enc{\ntostr\counter}
\\ & \tob^{3} &
\enc{\ntostr\counter} N_0N_1N_{\ems} (\la\varthree \succl\,\varthree) \cont 
\end{array}\]

Cases of $\counter$:
\begin{itemize}
\item \emph{Zero}, that is, $\counter =0 $, $\ntostr\counter = \ems$, and $\enc{\ntostr\counter} = \la{\var_0}{\la{\var_1}{\la{\var_\varepsilon}{\var_\varepsilon}}}$: then  
\[\begin{array}{llllllll}
\enc{\ntostr\counter} N_0N_1N_{\ems} (\la\varthree \succl\,\varthree) \cont 
& = & 
(\la{\var_0}{\la{\var_1}{\la{\var_\varepsilon}{\var_\varepsilon}}}) N_0N_1N_{\ems} (\la\varthree \succl\,\varthree) \cont \\ & \tob^{3} &
N_{\ems} (\la\varthree \succl\,\varthree) \cont
\\ & = &
(\la{\varf'} \la {\cont'} \cont'\,\enc{1\cons\ems}) \cont
\\
 & \tob & \cont\, \enc{1\cons\ems}
 \\ & = &  \cont\, \enc{\succf\, \ntostr 0}
\end{array}\]

\item \emph{Not zero}. Then there are two sub-cases, depending on the first character of the string $\ntostr\counter$:
\begin{itemize}
\item \emph{$0$ character}, \ie $\ntostr\counter = 0\cons\str$: then 
\[\begin{array}{rlllllll}
&&\enc{0\cons\strtwo} N_0N_1N_{\ems} (\la\varthree \succl\,\varthree) \cont \\
& = & 
(\la{\var_0}{\la{\var_1}{\la{\var_\varepsilon}\var_0\enc{\str}}}) N_0N_1N_{\ems} (\la\varthree \succl\,\varthree) \cont 
\\ & \tob^{3} &
N_{0} \enc{\str} (\la\varthree \succl\,\varthree) \cont 
\\ & = &
(\la{\varf'}\la {\str'}{\la {\cont'}{\appendchar1\cont' \str'}}) \enc{\str} (\la\varthree \succl\,\varthree) \cont 
\\ & \tob^{3} & 
 \appendchar1\cont \enc{\str}
\\ (\reflemmaeq{append-char}) & \tob^{\bigo{1}} &  
\cont\, \enc{1\cons\str}
\\ & = &  \cont\, \enc{\succf\, \ntostr\counter}
\end{array}\]

\item \emph{$1$ character}, \ie $\ntostr\counter = 1\cons\str$: then 
\[\begin{array}{rlllllll}
&&\enc{1\cons\strtwo} N_0N_1N_{\ems} (\la\varthree \succl\,\varthree) \cont \\
& = & 
(\la{\var_0}{\la{\var_1}{\la{\var_\varepsilon}\var_1\enc{\str}}}) N_0N_1N_{\ems} (\la\varthree \succl\,\varthree) \cont 
\\ & \tob^{3} &
N_{1} \enc{\str} (\la\varthree \succl\,\varthree) \cont 
\\ & = &
(\la{\varf'}\la {\str'}{\la {\cont'}{\varf'(\la\varthree\appendchar0\cont'\varthree) \str'}}) \enc{\str} (\la\varthree \succl\,\varthree) \cont 
\\ & \tob^{3} & 
(\la\varthree \succl\,\varthree)\,(\la\varthree\appendchar0\cont\varthree)\enc{\str}
\\ & \tob & 
 \succl\,(\la\varthree\appendchar0\cont\varthree)\enc{\str}
\\ (\ih) & \tob^{\bigo{\size\str}} &  
(\la\varthree\appendchar0\cont\varthree)\enc{\succf\,\str}
\\ & \tobdet &
\appendchar0\cont \enc{\succf\,\str}
\\ (\reflemmaeq{append-char}) & \tob^{\bigo{1}} &  \cont\, \enc{ 0\cons(\succf\,\str)}
\\ & = &  \cont\, \enc{\succf\, (1\cons\str)}
\\ & = &  \cont\, \enc{\succf\, \ntostr\counter}
\end{array}\]
\end{itemize}\qedhere
\end{itemize}
\end{proof}

\paragraph{Predecessor Function.} We now define and implement a predecessor 
function. We define it assuming that it shall only be applied to the enconding 
$\ntostr n$ of a natural number $n$ different from $0$, as it shall indeed be 
the case in the following. Such a predecessor function $\predf$ is defined as 
follows on the reversed binary representation (in Haskell-like syntax):
\[
\begin{array}{lllll}
	\predf &0\cons\str& = &1\cons(\predf\,\str)\\
	\predf &1\cons\ems& = &\ems\\
	\predf &1\cons b\cons\str& = &0\cons b\cons\str\\
\end{array}
\]
It is easily seen that $\predf (\ntostr n) = \ntostr{n-1}$ for all 
$0<n\in\nat$. Note that $\predf(\ntostr n)$ does not introduce a rightmost $0$ 
bit when it changes the rightmost bit of $\ntostr n$, that is, $\predf\, 001 = 
11$ and not $110$.

\begin{lemma}
	\label{l:pred}
	There is a $\l$-term $\predl$ such that for every continuation term $\cont$ and 
	every natural number $1\leq\counter\in\nat$, 
	$$
	\predl\, \cont  \enc{\ntostr\counter}
	\tobdet^{\bigo{\log\counter}} 
	\cont\, \enc{\predf\, \ntostr\counter}.
	$$
\end{lemma}

\begin{proof}
Define 
$\predl \defeq \fix\predaux$ and $\predaux \defeq \la\varf{\la {\cont'}{\la {\counter'}{\counter' N_0N_1N_\ems \varf\cont'}}}$ 
where:
\begin{itemize}
	\item $N_0\defeq \la {\strtwo'}\la\varf\la {\cont'}{\varf(\la\varthree\appendchar1\cont'\varthree)\strtwo'}$;
	\item $N_1\defeq \la {\strtwo'}\la\varf{{\strtwo'}M_{0}M_{1}M_\ems}$, where:
	\begin{itemize}
	\item $M_0\defeq \la v{\la k{\appendchar0(\la\varthree\appendchar1\cont\varthree) v}}$;
	\item $M_1\defeq \la v{\la k{\appendchar1(\la\varthree\appendchar1\cont\varthree) v}}$;
	\item $M_\ems\defeq \la {\cont'} \cont'\enc\ems$;
\end{itemize}

	\item $N_{\ems}$ is whatever closed term.
\end{itemize}
We rather prove $
	\predl\, \cont  \enc{\ntostr\counter}
	\tobdet^{\bigo{\size{\ntostr\counter}}} 
	\cont\, \enc{\predf\, \ntostr\counter}.
	$, where clearly $\size{\ntostr\counter} = \log \counter$, because the proof is naturally by induction on the length of $\ntostr\counter$ as a string. The first steps of the evaluation of $\predl\, \cont \enc{\ntostr\counter}$ are common to all natural numbers $1\leq \counter\in\nat$:
\[\begin{array}{llllllll}
\predl\, \cont \enc{\ntostr\counter} & = & \fix \predaux\, \cont \enc{\ntostr\counter}
\\ & \tob^{2} &
\predaux (\la\varthree \predl\, \varthree) \cont \enc{\ntostr\counter}
\\ & = &
(\la\varf{\la {\cont'}{\la {\counter'}{\counter' N_0N_1N_\ems \varf\cont'}}}) (\la\varthree \predl\, \varthree) \cont \enc{\ntostr\counter}
\\ & \tob^{3} &
\enc{\ntostr\counter} N_0 N_1N_\ems (\la\varthree \predl\, \varthree) \cont 
\end{array}\]

By hypothesis, $\counter\geq 1$. Then $\enc\counter$ is a non-empty string. Cases of its first character:
\begin{itemize}

\item \emph{$0$ character}, \ie ${\ntostr\counter} = 0\cons\strtwo$: then 
\[\begin{array}{rlllllll}
&&\enc{\ntostr\counter} N_0 N_1N_\ems (\la\varthree \predl\, \varthree) \cont 
\\
& = & 
(\la{\var_0}{\la{\var_1}{\la{\var_\varepsilon}\var_0\enc{\strtwo}}}) N_0 N_1N_\ems (\la\varthree \predl\, \varthree) \cont 
\\ 
& \tob^{3} &
N_0 \enc{\strtwo} (\la\varthree \predl\, \varthree) \cont
\\ & = &
(\la {\strtwo'}\la\varf\la {\cont'}{\varf(\la\varthree\appendchar1\cont'\varthree)\strtwo'}) \enc{\strtwo} (\la\varthree \predl\, \varthree) \cont
\\ & \tob^{3} & 
(\la\varthree \predl\, \varthree) (\la\varthree\appendchar1\cont'\varthree) \enc{\strtwo}
\\ & \tob & 
\predl (\la\varthree\appendchar1\cont\varthree) \enc{\strtwo}
\\ (\ih) & \tob^{\bigo{\size\strtwo}} &  
(\la\varthree\appendchar1\cont\varthree)\, \enc{\predf\, \strtwo}
\\  & \tob &  
\appendchar1\cont\, \enc{\predf\, \strtwo}
\\ (\reflemmaeq{append-char}) & \tob^{\bigo{1}} &  
\cont\, \enc{1\cons (\predf\, \strtwo)}
\\ & = &  \cont\, \enc{\predf\, {0\cons\strtwo}}
\\ & = &  \cont\, \enc{\predf\, {\ntostr\counter}}
\end{array}\]

\item \emph{$1$ character}, \ie ${\ntostr\counter} = 1\cons\strtwo$: then 
\[\begin{array}{rlllllll}
\enc{\ntostr\counter} N_0 N_1N_\ems (\la\varthree \predl\, \varthree) \cont 
& = & 
(\la{\var_0}{\la{\var_1}{\la{\var_\varepsilon}\var_1\enc{\strtwo}}}) N_0 N_1N_\ems (\la\varthree \predl\, \varthree) \cont 
\\ 
& \tob^{3} &
N_1 \enc{\strtwo} (\la\varthree \predl\, \varthree) \cont
	\\ & = &
	(\la {\strtwo'}\la\varf{{\strtwo'}M_{0}M_{1}M_\ems})\enc{\strtwo} (\la\varthree \predl\, \varthree) \cont\\
	&\tob& \enc{\strtwo}M_0M_1M_\ems \cont
	\end{array}\]
	
	There are three sub-cases, depending on the string $\strtwo$:
	\begin{itemize}
	\item $\strtwo$ is empty, \ie $\strtwo=\ems$. Then:
	\[\begin{array}{rlllllll}
	\enc{\ems}M_0M_1M_\ems \cont &=&
	(\la{\var_0}{\la{\var_1}{\la{\var_\varepsilon}{\var_\varepsilon}}})M_0M_1M_\ems k\\
	&\tob^3&M_\ems k
	\\ & = & 
	(\la {\cont'} \cont'\enc\ems)\cont
	\\  & \tob &  \cont \enc{ \ems}
		\\  & \tob &  \cont \enc{ \ntostr 0}
	\\ & = &  \cont \enc{\predf\, {\ntostr 1}}
\end{array}\]

\item \emph{$\strtwo$ start with $0$}, that is, 
$\ntostr\counter=1\cons\strtwo= 1\cons0\cons\strthree$. Then:
\[\begin{array}{rlllllll}
\enc{0\cons\strthree}\,M_0M_1M_\ems k&=&
(\la{\var_0}{\la{\var_1}{\la{\var_\varepsilon}\var_{0}\enc\strthree}})M_0M_1M_\ems k\\
&\tob^3&M_0 \enc\strthree k
\\
&=&(\la v{\la k{\appendchar0(\la\varthree\appendchar0\cont\varthree) v}})\enc\strthree k
\\&\tob^2&
\appendchar0 (\la\varthree\appendchar0\cont\varthree) \enc\strthree
\\ (\reflemmaeq{append-char}) & \tob^{\bigo{1}} &  (\la\varthree\appendchar0\cont\varthree) \enc{ 0\cons\strthree}
\\ & \tob &  \appendchar0\cont\, \enc{ 0\cons\strthree}
\\ (\reflemmaeq{append-char}) & \tob^{\bigo{1}} &  
\cont\, \enc{0\cons0\cons\strthree}
\\ & = &  \cont\, \enc{0\cons\strtwo}
\\ & = &  \cont\, \enc{\predf\, {1\cons\strtwo}}
\\ & = &  \cont\, \enc{\predf\, {\ntostr\counter}}
\end{array}\]

\item \emph{$\strtwo$ start with $1$}, that is, 
$\ntostr\counter=1\cons\strtwo= 1\cons1\cons\strthree$. Then:
\[\begin{array}{rlllllll}
\enc{0\cons\strthree}\,M_0M_1M_\ems k&=&
(\la{\var_0}{\la{\var_1}{\la{\var_\varepsilon}\var_{1}\enc\strthree}})M_0M_1M_\ems k\\
&\tob^3&M_1 \enc\strthree k
\\
&=&(\la v{\la k{\appendchar1(\la\varthree\appendchar0\cont\varthree) v}})\enc\strthree k
\\&\tob^2&
\appendchar1 (\la\varthree\appendchar0\cont\varthree) \enc\strthree
\\ (\reflemmaeq{append-char}) & \tob^{\bigo{1}} &  (\la\varthree\appendchar0\cont\varthree) \enc{ 1\cons\strthree}
\\ & \tob &  \appendchar0\cont\, \enc{ 1\cons\strthree}
\\ (\reflemmaeq{append-char}) & \tob^{\bigo{1}} &  
\cont\, \enc{0\cons1\cons\strthree}
\\ & = &  \cont\, \enc{0\cons\strtwo}
\\ & = &  \cont\, \enc{\predf\, {1\cons\strtwo}}
\\ & = &  \cont\, \enc{\predf\, {\ntostr\counter}}
\end{array}\]
\qedhere
\end{itemize}
\end{itemize}
\end{proof}

\paragraph{Lookup Function.} Given a natural number $\counter$, we need to be 
able to extract the $\counter+1$-th character from a non-empty string $\str$. 
The partial function $\lookupf$ can be defined as follows (in Haskell-like 
syntax):
\[
\begin{array}{llllll}
	\lookupf &\ntostr 0&(c\cons\str)& = &c\\
	\lookupf &\ntostr n&(c\cons\str)& = &\lookupf\ (\predf\, \ntostr n)\ \str & 
	\mbox{if }n>0\\
\end{array}
\]

\begin{lemma}\label{l:look}
	There is a $\l$-term $\lookupl$ such that for every continuation term $\cont$, 
	every natural number $n$ and every non-empty string $\inputstr\in\Bool^+$, 
	$$\lookupl\,k\enc{\ntostr n}\enc{\inputstr} \ \ \tobdet^{\bigo{n\log n}} \ 
	\ k\,\cods{\lookupf\, \ntostr n \inputstr}.$$
\end{lemma}

\begin{proof}

We can now code the function $\lookupl \defeq \fix \lookupaux$ where:
\[
\lookupaux \defeq \la\varf{\la {\cont'}{\la {\counter'}{\la {\inputstr'}{\counter'N_{0}N_{1}N_\ems \varf \cont' \inputstr'}}}}
\]
where:
\begin{itemize}
	\item $N_{0}\defeq \la{\strthree'}\la\varf\la{\cont'}\la{\inputstr'}\inputstr' M_{0} 
	M_{0} M_\ems \strthree' \varf \cont' $, where
	\begin{itemize}
	\item $M_{0}\defeq \la{\strtwo'}\la{\strthree'}\la\varf\la{\cont'} \appendchar0  (\la{\varthree''}\predl(\la{\varthree'}\varf\cont'\varthree' ) \varthree'') \strthree' \strtwo'$;
	\item $M_\ems$ is whatever closed term.
	\end{itemize}
	
	\item $N_{1}\defeq \la{\strthree'}\la\varf\la{\cont'}\la{\inputstr'}\inputstr' M_{1} M_{1} M_\ems  \strthree' \varf \cont'$, where
	\begin{itemize}
	\item $M_{1}\defeq \la{\strtwo'}\la{\strthree'}\la\varf\la{\cont'} \appendchar1  (\la{\varthree''}\predl(\la{\varthree'}\varf\cont'\varthree' ) \varthree'') \strthree' \strtwo'$;
	\item $M_\ems$ is whatever closed term.
	\end{itemize}
	
	\item $N_\ems\defeq \la\varf\la{\cont'}\la{\inputstr'}\inputstr' O_0 O_{1} 
	O_\stsym O_\ensym O_{\ems} \cont'$, where 
	
	\begin{itemize}
	\item $O_\bool \defeq \la{\str'}\la {\cont'} \cont' \cods{\bool}$;
	\item $O_\ems$ is whatever closed term.
	\end{itemize}
\end{itemize}
The first steps of the evaluation of $\lookupl\,\cont\enc{\ntostr \counter}\enc{\inputstr}$ are common to all strings $\inputstr\in\Bool^+$ and natural numbers $\counter\in\nat$:
\[\begin{array}{llllllll}
\lookupl\,\cont\enc{\ntostr \counter}\enc{\inputstr} & = & \fix \lookupaux\, k\enc{\ntostr \counter}\enc{\inputstr}
\\ & \tob^{2} &
\lookupaux (\la\varthree \lookupl\, \varthree) \cont\enc{\ntostr \counter}\enc{\inputstr}
\\ & = &
(\la\varf{\la {\cont'}{\la {\counter'}{\la {\inputstr'}{\counter'N_{0}N_{1}N_\ems \varf \cont' \inputstr'}}}}) (\la\varthree \lookupl\,\varthree)  \cont\enc{\ntostr \counter}\enc{\inputstr}
\\ & \tob^{4} &
\enc{\ntostr \counter} N_{0} N_{1} N_{\ems} (\la\varthree \lookupl\,\varthree) \cont \enc{\inputstr}
\end{array}\]
Cases of $n$:
\begin{itemize}
\item $n=0$, and so $\ntostr \counter = \ems$: then  
\[\begin{array}{llllllll}
&&\enc{\ems} N_{0}N_{1} N_{\ems} (\la\varthree \lookupl\,\varthree) \cont \enc\inputstr
\\ & = & 
(\la{\var_0}{\la{\var_1}{\la{\var_\varepsilon}{\var_\varepsilon}}}) N_{0}N_{1} N_{\ems} (\la\varthree \lookupl\,\varthree) \cont \enc\inputstr
{\var_\varepsilon}N_{\ems} (\la\varthree \lookupl\,\varthree) \cont \enc\inputstr
\\ & = &
(\la\varf\la{\cont'}\la{\inputstr'}\inputstr' O_0 O_{1} O_\stsym O_\ensym 
O_{\ems} \cont') (\la\varthree \lookupl\,\varthree) \cont \enc\inputstr
\\ & \tobdet^{3} &
\enc\inputstr O_0 O_{1} O_\stsym O_\ensym O_{\ems} \cont
\end{array}\]
Let $\inputstr$ start with $\bool\in \Bool$, that is, $\inputstr = \bool\cons\str$:
\[\begin{array}{llllllll}
\enc{\bool\cons\str}\, O_0 O_{1} O_\stsym O_\ensym O_{\ems} \cont
& = & 
(\la{\var_{0}}\la{\var_{1}}\la\var_{\ems}\var_{\bool} \enc\str) O_{0} O_{1} 
O_\stsym O_\ensym O_{\ems} \cont
\\ & \tobdet^{3} &
O_{\bool} \enc\str \cont
\\ & = &
(\la{\str'}\la {\cont'} \cont'\cods{\bool}) \enc\str \cont
\\ & \tob^{2} &
\cont \cods{\bool} 
\\ & = & 
\cont \cods{\lookupf\, \ems (\bool\cons\str)}
\\ & = & 
\cont \cods{\lookupf\, \ntostr 0 \istr}
\end{array}\]

\item \emph{Non-empty string starting with $0$}, that is, $n>0$ and $\ntostr \counter = 0 \cons \strthree$: then  
\[\begin{array}{llllllll}
&&\enc{\ntostr \counter} N_{0}N_{1}N_{\ems} (\la\varthree \lookupl\,\varthree) \cont \enc\inputstr
\\ & = & 
(\la{\var_0}{\la{\var_1}{\la{\var_\varepsilon} \var_{0} \enc{\strthree}}}) N_{0}N_{1}N_{\ems} (\la\varthree \lookupl\,\varthree) \cont \enc\inputstr
\\ & \tob^{3} &
N_{0} \enc{\strthree} (\la\varthree \lookupl\,\varthree) \cont \enc\inputstr
\\ & = &
( \la{\strthree'}\la\varf\la{\cont'}\la{\inputstr'}\inputstr' M_{0} 
	M_{0} M_\ems \strthree' \varf \cont' )\enc{\strthree} (\la\varthree \lookupl\,\varthree)\cont \enc\inputstr
\\ & \tob^{4} &
\enc\inputstr M_{0} M_{0} M_\ems \enc{\strthree} (\la\varthree \lookupl\,\varthree) \cont
\end{array}\]
Let $\inputstr$ start with $\bool\in \Bool$, that is, $\inputstr = \bool\cons\strtwo$:


\[\begin{array}{llllllll}
\enc{\bool\cons\strtwo} M_{0} M_{0} M_\ems \cont
\\
 = \\ 
(\la{\var_{0}}\la{\var_{1}}\la{\var_{\ems}} \var_{\bool} \enc\strtwo) M_{0} M_{0} M_\ems \enc{\strthree} (\la\varthree \lookupl\,\varthree) \cont
\\  \tobdet^{\size\alpone} \\
M_{0} \enc\strtwo\, \enc{\strthree} (\la\varthree \lookupl\,\varthree) \cont
\\  = \\
(\la{\strtwo'}\la{\strthree'}\la\varf\la{\cont'} \appendchar0  (\la{\varthree''}\predl(\la{\varthree'}\varf\cont'\varthree' ) \varthree'') \strthree' \strtwo') \enc\strtwo\, \enc{\strthree} (\la\varthree \lookupl\,\varthree)\cont
\\  \tobdet^{4} \\
\appendchar0  (\la{\varthree''}\predl(\la{\varthree'}(\la\varthree \lookupl\,\varthree)\cont\varthree' ) \varthree'')  \enc\strthree\, \enc\strtwo
\\ (\reflemmaeq{append-char}  \tobdet^{\bigo{1}} \\
(\la{\varthree''}\predl(\la{\varthree'}(\la\varthree \lookupl\,\varthree)\cont\varthree' ) \varthree'') \enc{0\cons\strthree}\, \enc\strtwo
\\  \tobdet \\
\predl(\la{\varthree'}(\la\varthree \lookupl\,\varthree)\cont\varthree' ) \enc{0\cons\strthree}\, \enc\strtwo
\\  = \\
\predl(\la{\varthree'}(\la\varthree \lookupl\,\varthree)\cont\varthree' ) \enc{\ntostr \counter}\, \enc\strtwo
\\ (\reflemmaeq{pred})  \tobdet^{\bigo{\log \counter}} \\
(\la{\varthree'}(\la\varthree \lookupl\,\varthree)\cont\varthree' ) \enc{\predf\, \ntostr \counter}\, \enc\strtwo
\\  \tobdet^{2} \\
\lookupl\, \cont\, \enc{\predf\,\ntostr \counter}\, \enc\strtwo
\\  = \\
\lookupl\, \cont\, \enc{\ntostr{ n-1}}\, \enc\strtwo
\\ (\ih)  \tobdet^{\bigo{(n-1)\cdot\log{(n-1)}}} \\
\cont (\cods{\lookupf\, \ntostr{n-1} \strtwo})
\\  = \\
\cont (\cods{\lookupf\, \ntostr \counter \str})
\end{array}\]

The number of $\beta$ steps then is $\bigo{\log\counter} + \bigo{(\counter-1)\cdot\log{(\counter-1)}} + h$ for a certain constant $h$, which is bounded by $\bigo{\counter\cdot\log{\counter}}$, as required.

\item \emph{Non-empty string starting with $1$}, that is, $n>0$ and $\ntostr \counter = 1 \cons \strthree$: same as the previous one, simply replacing $N_{0}$ with $N_{1}$, and thus $M_{0}$ with $M_{1}$. In particular, it takes the same number of steps.\qedhere
\end{itemize}
\end{proof}


\subsection{The New Encoding of Turing Machines}
\paragraph{Turing Machines.}
Let $\Boolin \defeq \set{0,1,\stsym,\ensym}$ and $\Boolb \defeq 
\set{0,1,\elemblank}$ where $\stsym$ and $\ensym$ delimit the input (binary) 
string, and  
$\elemblank$ is our notation for the blank symbol. 
A deterministic binary Turing machine $\mathcal{M}$ \emph{with input} is a 
tuple 
$(\States,\statein,\statefint,\statefinf,\delta)$ consisting of:
\begin{itemize}
	\item
	A finite set $Q=\{\state_1,\ldots,\state_m\}$ of \emph{states};
	\item
	A distinguished state $\statein\in Q$, called the \emph{initial
		state};
	\item
	Two distinguished states $\Statesfin\defeq\set{\statefint,\statefinf} 
	\subseteq \States$, 
	called the 
	\emph{final
		states}; 
	\item
	A partial \emph{transition function} 
	$\delta:\Boolin\times\Boolb\times\States\rightharpoonup 
	\{-1,+1,0\}\times\Boolb\times\{\leftarrow,\rightarrow,\downarrow\} \times 
	\States$
	such that
	$\delta(\bool,\elem,\state)$ is defined only if $\state\notin 
	\Statesfin$.
\end{itemize}
A configuration for $\mathcal{M}$ is a tuple
$$(\inputstr,\counter,\wstrl, \elem, \wstrr, \state) \in 
\Boolin^*\times\Nset\times\Boolb^*\times\Boolb\times\Boolb^*\times \States$$
where:
\begin{itemize}
	\item $\inputstr$ is the 
	immutable input string and is formed as 
	$\inputstr=\stsym\cons\str\cons\ensym,\ \str \in \Bool^*$;
	\item $\counter\in \Nset$ represents the position of the input head. It is 
	meant to be represented in binary (that is, as an element of $\Bool^*$), to 
	take space $\log \counter$, but for ease of reading we keep referring to it 
	as a number rather than as a string;
	\item $\wstrl \in \Boolb^*$ is the work tape on the left of the work head;
	\item $\elem \in \Boolb$ is the element on the cell of the work tape read 
	by the work 
	head;
	\item $\wstrr \in \Boolb^*$ is the work tape on the right of the work head;
	\item $\state \in \States$ is the state of the machine.
\end{itemize}
For readability, we usually write a configuration $(\inputstr,\counter,\wstrl, 
\elem, \wstrr, \state)$ as $(\inputstr,\counter \csep \wstrl, \elem , \wstrr 
\csep \state) $, separating the input components, the working components, and 
the current state.

Given an input string $\inputstr \in\Boolin^*$ (where 
$i=\stsym\cons\str\cons\ensym$ and $\str\in\Bool^*$) we define:
\begin{itemize}
	\item the \emph{initial configuration} $\initconfigs$
	for $\inputstr$ is $\initconfigs \defeq 
	(\inputstr,0 \csep \varepsilon,\elemblank,\varepsilon 
	\csep \statein)$,
	\item the \emph{final configuration}
	$\config_{\tt fin} \defeq 
	(\str,\counter \csep \wstrl,\elem,\wstrr \csep\state)$, where 
	$\state\in\Statesfin$.	
\end{itemize}
For readability, a transition, say, $\delta(\inputstr_{\counter},\elem,\state) 
= 
(-1,\elem',\leftarrow, \state')$, is usually written as $(-1 \csep 
\elem',\leftarrow \csep \state')$ to stress the three components corresponding 
to those of configurations (input, work, state).

As in Goldreich, we assume that the machine never scans the input beyond the 
boundaries of the input. This does not affects space complexity.

\emph{An example of transition}: if $\delta(\inputstr_{\counter},\elem,\state) 
= 
(-1 \csep \elem',\leftarrow \csep \state')$, then $\M$ 
evolves from $\config = (\inputstr,n \csep \wstrl\elem'',\elem,\wstrr \csep 
\state)$, 
where the $\counter$th character of $\inputstr$ is $\inputstr_{\counter}$, 
to 
$\configtwo = (\inputstr,n-1\csep \wstrl,\elem'',\elem'\wstrr \csep \state')$ 
and if 
the tape on the 
left of the work head is empty, \ie if $\config = 
(\inputstr,n \csep \varepsilon,\elem,\wstrr \csep \state)$, then the content of 
the 
new head cell 
is a blank symbol, that is, $\configtwo \defeq 
(\inputstr,n-1 \csep \varepsilon,\elemblank,\elem'\wstrr \csep \state')$. The 
same happens if the tape on the right of the work head is empty. If $\M$ has a 
transition 
from $\config$ to $\configtwo$ we write $\config \tomachtur \configtwo$. A 
configuration having as state a final state $\state\in\Statesfin$ is 
\emph{final} and cannot evolve.

A Turing machine $(\States,\statein,
\statefint,\statefinf,\delta)$ computes the function 
$f:\Bool^*\rightarrow\Bool$ 
in time $T:\Nset\rightarrow\Nset$ and space  $S:\Nset\rightarrow\Nset$
if for every $\istr\in\Bool^+$, the
initial configuration for $\istr$ evolves to a final configuration of state 
$\state_{f(\istr)}$ in $T(\size\istr)$ steps and using at most $S(\size\istr)$ 
cells on the work tape.

\paragraph{Encoding configurations.} A configuration 
$(\inputstr,\counter \csep \strone,\elem,\strtwo \csep\state)$ of a machine
$\M = (\States,\statein, \statefint,\statefinf,\delta)$ is 
represented by
the term
$$
\encp{(\inputstr,\counter \csep \wstrl,\elem,\wstrr \csep \state)}{\M}
\defeq
\lambda x. (x \encp{\inputstr}{\Bool^+}\encp{\ntostr \counter}{\Bool}\;
\encp{\wstrl^{\rop}}{\Boolb^*}\;\cod{\elem}{\Boolb}\;\encp{\wstrr}{\Boolb^*}\;\cod{\state}{\States}).
$$
where $\wstrl^{\rop}$ is the string $\wstrl$ with the elements in reverse 
order. 
We shall often rather write
$$
\enc{(\inputstr,\counter \csep \wstrl,\elem,\wstrr \csep \state)}
\defeq
\lambda x. (x\; \enc{\inputstr}\; \enc{\ntostr \counter}\;
\enc{\wstrl^{\rop}}\;\cods\elem \; \enc\wstrr \; \cods\state).
$$
letting the superscripts implicit. To ease the reading, we sometimes use the 
following  notation for tuples 
$\tuple{\tmfive,\tmsix \csep \tm,\tmtwo,\tmthree \csep \tmfour} \defeq 
\lambda x. (x \tmfive\tmsix\tm \tmtwo \tmthree \tmfour)$, so that 
$\enc{(\inputstr,\counter \csep \wstrl,\elem,\wstrr \csep \state)} = 
\tuple{\enc{\inputstr},\enc{\ntostr \counter} \csep \enc{\wstrl^{\rop}}, 
\cods\elem, 
	\enc\wstrr \csep \cods\state}$.


\paragraph{Encoding the transition functions} The transition function $\delta(\bool,\elem,\state)$ is implemented by looking up into a 3-dimensional table $T$ having for coordinates:
\begin{itemize}
\item \emph{Input}: the current bit $\bool$ on the input tape, which is actually retrieved 
from the input tape $i$ and the counter $n$ of the current input position,
\item \emph{Work}: the current character $\elem$ on the work tape, and 
\item \emph{State}: the current state $\state$, 
\end{itemize}
The transition function is encoded as a recursive $\l$-term $\transs$ taking as argument the encodings of $\istr$, and $\counter$---to retrieve $\bool$---and $\elem$ and $\state$. It works as follows:
\begin{itemize}
\item It first retrieves $\bool$ from $\counter$ and $\istr$ by applying the $\lookupl$ function;
\item It has a subterm $A_{\bool}$ for the four values of $\bool$. The right 
sub-term is selected by applying the encoding $\cods\bool$ of $\bool$ to 
$A_{0},A_{1},A_\stsym$ and $A_\ensym$.

\item Each $A_{\bool}$ in turn has a sub-term $B_{\bool,\elem}$ for every character $\elem\in \Boolb$, corresponding to the working tape coordinates. The right sub-term is selected by applying the encoding $\cods\elem$ of  the current character $\elem$ on the work tape to $B_{\bool,0}, B_{\bool,1}, B_{\bool, \elemblank}$.

\item Each $B_{\bool,\elem}$ in turn has a subterm $C_{\bool,\elem,\state}$ for every character $\state$ in $\States$. The right sub-term is selected by applying the encoding $\cods\state$ of  the current state $\state$ to $C_{\bool,\elem,\state_{1}}, \ldots, B_{\bool, \elem,\state_{\size\States}}$.

\item The subterm $C_{\bool,\elem,\state}$ produces the (encoding of the) next configuration according to the transition function $\delta$. If $\delta$ decreases (resp. increases) the counter for the input tape then $C_{\bool,\elem,\state}$ applies $\predl$ (resp. $\succl$) to the input counter and then applies a term corresponding to the required action on the work tape, namely:
\begin{itemize}
	\item $S$ (for \emph{stay}) if the head does not move. This case is easy, $S$ simply produces the next configuration.
	\item $L$ if it moves left. Let $\wstrl = \wstr\elem''$, $\elem'$ the element that the transition has to write and $\state'$ the new state. Then $L$ has a subterm $L_{\elem''}^{\elem',\state'}$ for each $\elem''\in\Boolb$ the task of which is to add $\elem'$ to the right part of the work tape, remove $\elem''$ from the left part of the work tape (which becomes $\wstr$), and make $\elem''$ the character in the work head position.
	\item $R$ if it moves right. Its structure is similar to the one of $L$.
\end{itemize}
In order to be as modular as possible we use the definition of $S$, $L$, and $R$ for the cases when the input head moves also for the cases where it does not move, even if this requires a useless (but harmless) additional update of the counter $\counter$.
\end{itemize}

Define the term 
	$\trans{\M}$, or 
	simply 
	$\transs$, as follows.
\[\begin{array}{rclll}
\transaux 
& \defeq & 
\lambda \var.\la{\cont'}\lambda \config'.\config'(\la{\istr'}\la{\counter'}\la{\wstrl'}\la{\elem'}\la{\wstrr'}\la{\state'}\lookupl\,K \istr' \counter')
\\

\\
\transs 
& \defeq &
\fix \transaux,
\end{array}\]
where:
\[\begin{array}{rcl}
T & \defeq & \la {\bool'} \bool'\, 
A_{0}A_{1}A_{\stsym}A_{\ensym}\elem'\state'\var \cont' \istr' \counter'\wstrl' 
\wstrr'
\\
A_{\bool} & \defeq & \la{\elem'} \elem' B_{\bool,0} B_{\bool,1} B_{\bool,\elemblank} 
\\\\
B_{\bool,\elem} & \defeq & \la{\state'}\state' C_{\bool,\elem,\state_{1}}\ldots C_{\bool,\elem,\state_{\size\States}}
\\

C_{\bool,\elem,\state}&\defeq&
\la\var \la{\cont'} \la{\istr'} \la{\counter'} \la{\wstrl'} \la{\wstrr'}
    \left\{
   \begin{array}{ll}
	\cont' \tuple{\istr',\counter' \csep \wstrl',\cods{\elem}, 
   \wstrr' \csep \cods{\state} }
      \hfill\mbox{if $\state\in\Statesfin$}\\[3pt]
      
 S \counter' 
      \hfill\mbox{\qquad if 
      $\delta(\bool,\elem,\state)=(0 \csep \elem',\downarrow \csep \state')$}\\[3pt]
      
   L \counter' 
      \hfill \mbox{if 
      $\delta(\bool,\elem,\state)=(0\csep \elem',\leftarrow \csep \state')$}
      \\[3pt]

   R \counter' 
      \hfill \mbox{if 
      $\delta(\bool,\elem,\state)=(0\csep \elem',\rightarrow \csep \state')$}\\[3pt]
  
 \predl  S \counter'
  \hfill\mbox{\qquad if 
  	$\delta(\bool,\elem,\state)=(-1\csep \elem',\downarrow \csep \state')$}\\[3pt]
  
\predl L \counter'
  \hfill \mbox{if 
  	$\delta(\bool,\elem,\state)=(-1\csep \elem',\leftarrow \csep \state')$
	}
	\\[3pt]

  \predl R \counter'
  \hfill \mbox{if 
  	$\delta(\bool,\elem,\state)=(-1\csep \elem',\rightarrow \csep \state')$}

   \\[3pt]
   
   \succl S \counter'
  \hfill\mbox{\qquad if 
  	$\delta(\bool,\elem,\state)=(+1\csep \elem',\downarrow \csep \state')$}\\[3pt]
  
\succl L \counter'
  \hfill \mbox{if 
  	$\delta(\bool,\elem,\state)=(+1\csep \elem',\leftarrow \csep \state')$
	}
	\\[3pt]

  \succl  R \counter'
  \hfill \mbox{if 
  	$\delta(\bool,\elem,\state)=(+1\csep \elem',\rightarrow \csep \state')$}
   \end{array}
    \right.
    
    \\[3pt]
    S & \defeq & \la {\counter''} \var\cont' \tuple{\istr',\counter''\csep 
  	\wstrl',\cods{\elem'}, \wstrr' \csep \cods{\state'}} 
    \\
    L & \defeq & \la{\counter''}\wstrl' L_0^{\state',\elem'} L_1^{\state',\elem'} L_{\elemblank}^{\state',\elem'}L_{\ems}^{\state',\elem'} \var \cont' \istr' \counter'' \wstrr'
    \\
    R & \defeq & \la{\counter''}\wstrr' R_0^{\state',\elem'} R_1^{\state',\elem'} R_{\elemblank}^{\state',\elem'}R_{\ems}^{\state',\elem'} \var \cont' \istr' \counter'' \wstrl'
    \\
    L_{\elem''}^{\state',\elem'}
      & \defeq &
      \la{\wstrl'} \la\var \la{\cont'} \la{\istr'} \la{\counter'}\appendchar{\elem'}(\lambda \wstrr'.\var\cont' 
\tuple{\istr',\counter' \csep \wstrl',\cods{\elem''},\wstrr' \csep \cods{\state'}})
\\[3pt]
    L_{\ems}^{\state',\elem'}
      & \defeq &
\la\var \la{\cont'} \la{\istr'} \la{\counter'}
       \appendchar{\elem'}((\la d \lambda \wstrr'.\var\cont' 
      \tuple{\istr',\counter'\csep d,\cods{\elemblank},\wstrr' \csep  
      \cods{\state'}})\enc\ems)
      \\[3pt]
          R_{\elem''}^{\state',\elem'}
           & \defeq &
      \la{\wstrr'} \la\var \la{\cont'} \la{\istr'} \la{\counter'}\appendchar{\elem'}(\lambda \wstrl'.\var\cont' 
\tuple{\istr',\counter' \csep \wstrl',\cods{\elem''},\wstrr' \csep \cods{\state'}})
\\[3pt]
    R_{\ems}^{\state',\elem'}
      & \defeq &
\la\var \la{\cont'} \la{\istr'} \la{\counter'}
       \appendchar{\elem'}((\la d \lambda \wstrl'.\var\cont' 
      \tuple{\istr',\counter'\csep \wstrl',\cods{\elemblank},d \csep  
      \cods{\state'}})\enc\ems)

\end{array}\]


\section{Time Correctness of the Encoding}
\paragraph{Turning the input string into the initial configuration.} The 
following lemma provides the term $\inits $ that builds the initial 
configuration.

\begin{lemma}[Turning the input string into the initial configuration]
	\label{l:init-config}
	Let $\M=(\States,\statein,
	\statefint,\statefinf,\delta)$ be a Turing machine. There is a term 
	$\init\M{}$, or 
	simply $\inits$, such that for every continuation term $\cont$ and for every input string $\inputstr\in\Boolin^*$ 
	(where 
	$i=\stsym\cons\str\cons\ensym$ and $\str\in\Bool^*$):
	$$\inits\, \cont\, \enc{\inputstr} \ \ \tobdet^{\Theta(1)} \ \
	\cont\, \enc\initconfigs$$ 
	where $\initconfigs$ is the initial configuration of $\M$ for $\inputstr$.
\end{lemma}

\begin{proof}
	Define
	\[
		\inits \defeq (\la {d} \la {e} \la {f} \la{\cont'} \la{\inputstr'} 
		\cont'
		\tuple{\inputstr', d \csep e ,
			\cod{\elemblank}{\Boolb}, f \csep 
			\cod{\statein}{\States}})\enc{\ntostr 
			0}\encp{\varepsilon}{\Boolb^*} \encp{\varepsilon}{\Boolb^*}
	\]
	Please note that the term is not in normal form. This is for technical 
	reasons that will be clear next.
	Then 
	\[\begin{array}{rclcl}
	\inits\, \cont\, \encp\inputstr{\Boolin^*}
	& = &
	(\la {d} \la {e} \la {f} \la{\cont'} \la{\inputstr'} 
	\cont'
	\tuple{\inputstr', d \csep e ,
		\cod{\elemblank}{\Boolb}, f \csep 
		\cod{\statein}{\States}})\enc{\ntostr 
		0}\encp{\varepsilon}{\Boolb^*} \encp{\varepsilon}{\Boolb^*} \cont 
		\encp\inputstr{\Boolin^*}\\
	& \tobdet^5 &
	\cont\, 
	\tuple{\encp{\inputstr}{\Boolin^*},\enc{\ntostr 0} \csep 
	\encp{\varepsilon}{\Boolb^*},
		\cod{\elemblank}{\Boolb}, \encp{\varepsilon}{\Boolb^*} \csep
		\cod{\statein}{\States}}\\
	& = &
	\cont\, \enc{(\inputstr,0 \csep \varepsilon, \elemblank, \ems \csep \statein)}\\
	& = &
	\cont\, \enc\initconfigs
	\end{array}\]\qedhere
\end{proof}

\paragraph{Extracting the output from the final configuration.}

\begin{lemma}[Extracting the output from the final configuration]
	\label{l:final-config}
	Let $\M=(\States,\statein,
	\statefint,\statefinf,\delta)$ be a Turing machine. There is 
	a term $\final\M{}$, or simply $\finals$, such that for every continuation term $\cont$and  for every final
	configuration 
	$\config$ of state $\state\in\Statesfin$:
	\[\finals\, \cont\, \enc\config \tobdet^{\Theta(\size\States)} 
	\begin{cases}
		\cont(\la\var{\la\vartwo\var)}&\text{ if }\state= \statefint\\
		\cont(\la\var{\la\vartwo\vartwo)} &\text{ if }\state= \statefinf
	\end{cases}\]
\end{lemma}

\begin{proof}
	Define
	\begin{align*}
		\finals &\defeq \la {\cont'} \la {\config'}\config'(\lambda \inputstr'.\lambda n'.\lambda 
		\wstrl'.\lambda \elem'.\lambda 
		\wstrr'.\lambda \state'. \state' N_1\ldots N_{\size{\States}}\cont')
	\end{align*}
	where:
	\[
	N_i\defeq\begin{cases}
	\la{\cont'}{\cont'(\la\var{\la\vartwo\var)}} &\text{ if }q_i= \statefint\\
	\la{\cont'}{\cont'(\la\var{\la\vartwo\vartwo)}} &\text{ if }q_i= \statefinf\\
	\mbox{whatever closed term (say, the identity)} &\text{ otherwise }
	\end{cases}
	\]
	Then:
	\[\begin{array}{rclcl}
&&	\finals\, \cont\, \enc\config\\
	& = &
	(\la {\cont'} \la {\config'}\config'(\lambda \inputstr'.\lambda n'.\lambda 
		\wstrl'.\lambda \elem'.\lambda 
		\wstrr'.\lambda \state'. \state' N_1\ldots N_{\size{\States}}\cont')
) \cont \enc\config\\
	& \tobdet^2 &
	\enc\config(\lambda \inputstr'.\lambda n'.\lambda 
		\wstrl'.\lambda \elem'.\lambda 
		\wstrr'.\lambda \state'. \state' N_1\ldots N_{\size{\States}}\cont)\\
	& = &
	\enc{ (\inputstr,\counter \csep \wstrl, \elem, \wstrr \csep \state) }(\lambda \inputstr'.\lambda n'.\lambda 
		\wstrl'.\lambda \elem'.\lambda 
		\wstrr'.\lambda \state'. \state' N_1\ldots N_{\size{\States}}\cont)\\
	& = &
	(\lambda 
	x.x\encp{\inputstr}{\Bool^*}\enc{\ntostr\counter}\; 
	\encp{\wstrl^{\rop}}{\Boolb^*}\; 
	\cod{\elem}{\Boolb}\;\encp{\wstrr}{\Boolb^*}\;\cod{\state}{\States})
	 (\lambda \inputstr'.\lambda n'.\lambda 
		\wstrl'.\lambda \elem'.\lambda 
		\wstrr'.\lambda \state'. \state' N_1\ldots N_{\size{\States}}\cont)\\
	& \tobdet &
	(\lambda \inputstr'.\lambda n'.\lambda 
		\wstrl'.\lambda \elem'.\lambda 
		\wstrr'.\lambda \state'. \state' N_1\ldots N_{\size{\States}}\cont) 
	\encp{\inputstr}{\Bool^*}\enc{\ntostr\counter}\;\encp{\wstrl^{\rop}}{\Boolb^*} 
	\; 
	\cod{\elem}{\Boolb}\;\encp{\wstrr}{\Boolb^*}\;\cod{\state}{\States}\\
	 
	&\tobdet^6 &
	\cod{\state}{\States} N_1\ldots N_{\size{\States}}\cont
	\\&=&
	(\la{\var_1\ldots\var_{\size\States}}{\var_j})N_1\ldots N_{\size{\States}}\cont\\ 
	&\tobdet^{\size\States}&N_j\cont
	\end{array}\]
	If $\state=\statefint$, then:
	\[\begin{array}{rclcl}
	N_j\cont & = & (\la{\cont'}{\cont'(\la\var{\la\vartwo\var})})\cont\\
	&\tobdet&\cont(\la\var{\la\vartwo\var})
	\end{array}\]
	If $\state=\statefinf$, then:
	\[\begin{array}{rclcl}
	N_j\cont & = & (\la{\cont'}{\cont'(\la\var{\la\vartwo\vartwo})})\cont\\
	&\tobdet&\cont(\la\var{\la\vartwo\vartwo})
	\end{array}\]\qedhere
\end{proof}

\paragraph{Simulation of a machine transition.} Now we show that the given encoding of the 
transition function $\delta$ of a Turing machine as a $\l$-term simulates every single transition in $\bigo{\size\istr\log{\size\istr}}$ time, where $\istr$ is the input string. This is the heart of the 
encoding, and the most involved proof.

\begin{lemma}[Simulation of a machine transition]
	\label{l:trans-sim}
	Let $\M=(\States,\statein,\statefint,\statefinf,\delta)$ be a Turing 
	machine. The term 
	$\trans{\M}$ is
	such that for every continuation term $\cont$ and for every configuration 
	$\config$ of input string $\istr\in\Bool^{+}$: 
	\begin{itemize}
		\item \emph{Final configuration}:
		if $\config$ is a final configuration
		then $\transs\, \cont\, 
		\enc\config \tobdet^{\bigo{\size\istr\log{\size\istr}}} \cont\, 
		\enc\config$;
		\item \emph{Non-final configuration}: if $\config \tomachtur 
		\configtwo$ then 
		$\transs\, \cont\, 
		\enc\config 
		\tobdet^{\bigo{\size\istr\log{\size\istr}}} \transs\, \cont\, 
		\enc\configtwo$.
	\end{itemize}
\end{lemma}

\begin{proof}

Let $C=(\istr,\counter \csep \wstrl,\elem,\wstrr \csep \state)$. We are now 
going to show the details of how the $\l$-calculus simulates the transition 
function. At the level of the number of steps, the main cost is payed at the 
beginning, by the $\lookupl$ function that looks up the $\counter$-th character 
of the input string $\istr$. The cost of one such call is 
$\bigo{\counter\log\counter}$, but since $n$ can vary and $n\leq \size\istr$, 
such a cost is bound by $\bigo{\size\istr\log{\size\istr}}$. The cases of 
transition where the position on the input tape does not change have a constant 
cost. Those where the input position changes require to change the counter 
$\counter$ via $\predl$ or $\succl$, which requires $\bigo{\log \counter}$, 
itself bound by the cost $\bigo{\size\istr\log{\size\istr}}$ of the previous 
look-up.

Now, if $\lookupf\, \ntostr n \istr = \bool$ then:
\begin{center}
$\begin{array}{llllll}
\transs\,  \cont\,  \enc\config \\
=\\ 
\fix \transaux \cont  \enc\config 
\\  \tobdet^2 \\
\transaux (  \la\varthree\fix \transaux\varthree ) \cont  \enc\config
\\ =\\
\transaux ( \la\varthree\transs\,\varthree  ) \cont  \enc\config
\\ =\\
(\lambda \var.\la{\cont'}\lambda \config'.\config'(\la{\istr'}\la{\counter'}\la{\wstrl'}\la{\elem'}\la{\wstrr'}\la{\state'}\lookupl\,T \istr' \counter')) ( \la\varthree\transs\,\varthree  ) \cont  \enc\config
\\  \tobdet^3 \\
\enc\config(\la{\istr'}\la{\counter'}\la{\wstrl'}\la{\elem'}\la{\wstrr'}\la{\state'}\lookupl\,T\isub\var{\la\varthree\transs\,\varthree} \istr' \counter')
\\ =\\
\enc{(\istr,\counter\csep \wstrl,\elem,\wstrr \csep \state)} (\la{\istr'}\la{\counter'}\la{\wstrl'}\la{\elem'}\la{\wstrr'}\la{\state'}\lookupl\,T\isub\var{\la\varthree\transs\,\varthree} \istr' \counter')
\\ =\\
(\lambda x. x\enc{\istr}\enc{\ntostr\counter}
\enc{\wstrl^\rop} \;\cods{\elem} \;\enc\wstrr \;\cods{\state} )
 (\la{\istr'}\la{\counter'}\la{\wstrl'}\la{\elem'}\la{\wstrr'}\la{\state'}\lookupl\,T\isub\var{\la\varthree\transs\,\varthree} \istr' \counter')

\\ =\\
(\la{\istr'}\la{\counter'}\la{\wstrl'}\la{\elem'}\la{\wstrr'}\la{\state'}\lookupl\,T\isub\var{\la\varthree\transs\,\varthree} \istr' \counter')
 \enc{\istr}\enc{\ntostr\counter}
\enc{\wstrl^\rop} \;\cods{\elem} \;\enc\wstrr \;\cods{\state}

\\  \tobdet^{6} \\
\lookupl\,(\la {\bool'} \bool'\, 
A_{0}A_{1}A_{\stsym}A_{\ensym}\cods\elem\, \cods\state 
(\la\varthree\transs\,\varthree) \cont' \enc\istr \enc{\ntostr\counter} 
\enc{\wstrl^{\rop}} \enc{\wstrr}) \enc\istr \enc{\ntostr\counter}

 \\ \reflemmaeq{look}  \tobdet^{\bigo{\counter\log\counter}} \\
 (\la {\bool'} \bool'\, 
A_{0}A_{1}A_{\stsym}A_{\ensym}\cods\elem\, \cods\state 
(\la\varthree\transs\,\varthree) \cont' \enc\istr \enc{\ntostr\counter} 
\enc{\wstrl^{\rop}} \enc{\wstrr}
)\cods\bool 

 \\  \tobdet \\
\cods\bool\, 
A_{0}A_{1}A_{\stsym}A_{\ensym} \cods{\elem}\, \cods{\state} ( 
\la\varthree\transs\,\varthree) \cont \enc{\istr}\enc{\ntostr\counter}
\enc{\wstrl^\rop}  \enc\wstrr

 \\  \tobdet^{4} \\
A_{\bool} \cods{\elem}\,\cods{\state} ( \la\varthree\transs\,\varthree) \cont \enc{\istr}\enc{\ntostr\counter}
\enc{\wstrl^\rop} \;\enc\wstrr 

 \\ =\\
(\la{\elem'} \elem' B_{\bool,0} B_{\bool,1} B_{\bool,\elemblank} ) \cods{\elem}\,\cods{\state} ( \la\varthree\transs\,\varthree) \cont \enc{\istr}\enc{\ntostr\counter}
\enc{\wstrl^\rop} \;\enc\wstrr 

 \\  \tobdet \\
 \cods\elem B_{\bool,0} B_{\bool,1} B_{\bool,\elemblank} \cods\state ( \la\varthree\transs\,\varthree) \cont \enc\istr \enc{\ntostr\counter} \enc{\wstrl^{\rop}} \enc\wstrr

  \\  \tobdet^{3} \\
  B_{\bool,\elem} \cods\state ( \la\varthree\transs\,\varthree) \cont \enc\istr \enc{\ntostr\counter} \enc{\wstrl^{\rop}} \enc\wstrr

    \\ =\\
  (\la{\state'}\state' C_{\bool,\elem,\state_{1}}\ldots C_{\bool,\elem,\state_{\size\States}}) \cods\state ( \la\varthree\transs\,\varthree) \cont \enc\istr \enc{\ntostr\counter} \enc{\wstrl^{\rop}} \enc\wstrr

    \\  \tobdet \\
  \cods\state C_{\bool,\elem,\state_{1}}\ldots C_{\bool,\elem,\state_{\size\States}}  ( \la\varthree\transs\,\varthree) \cont \enc\istr \enc{\ntostr\counter} \enc{\wstrl^{\rop}} \enc\wstrr

    \\  \tobdet^{\size\States} \\
    C_{\bool,\elem,\state} ( \la\varthree\transs\,\varthree) \cont \enc\istr \enc{\ntostr\counter} \enc{\wstrl^{\rop}} \enc\wstrr

\end{array}$
 \end{center}
  
 Now, consider the following 
 four cases, depending on the value of $\delta(\bool,\elem,\state)$:
 \begin{enumerate}
    \item \emph{Final state}: 
    if $\delta(\bool,\elem,\state)$ is undefined, then $\state \in\Statesfin$ and replacing $C_{\bool,\elem,\state}$ with the corresponding $\l$-term we obtain:
\begin{center}
$\begin{array}{llllll}
C_{\bool,\elem,\state} ( \la\varthree\transs\,\varthree) \cont \enc\istr \enc{\ntostr\counter} \enc{\wstrl^{\rop}} \enc\wstrr
\\	
=\\
    (\la\var \la{\cont'} \la{\istr'} \la{\counter'} \la{\wstrl'} \la{\wstrr'} \cont' \tuple{\istr',\counter' \csep \wstrl',\cods{\elem}, \wstrr' \csep \cods{\state} }) ( \la\varthree\transs\,\varthree) \cont \enc\istr \enc{\ntostr\counter} \enc{\wstrl^{\rop}} \enc\wstrr
    \\
     \tobdet^6 \\
   \cont  \tuple{\enc\istr,\enc{\ntostr\counter} \csep \enc{\wstrl^{\rop}} ,\cods{\elem}, \enc\wstrr \csep \cods{\state} }\\
    =\\
    \cont  \enc{(\istr,\counter \csep \wstrl ,\elem,\wstrr \csep \state)}\\
    =\\
    \cont  \enc\config\\
\end{array}$
\end{center}
 
    \item \emph{The heads do not move}: 
    if 
    $\delta(\bool,\elem,\state)=(0 \csep \elem',\downarrow \csep \state')$,
     then $\configtwo = (\istr,\counter \csep \wstrl, \elem',\wstrr,\state')$. The simulation continues as follows:
\begin{center}
$\begin{array}{llllll}
C_{\bool,\elem,\state} ( \la\varthree\transs\,\varthree) \cont \enc\istr \enc{\ntostr\counter} \enc{\wstrl^{\rop}} \enc\wstrr
\\	
     = \\
        (\la\var \la{\cont'} \la{\istr'} \la{\counter'} \la{\wstrl'} \la{\wstrr'} 
S \counter')
   ( \la\varthree\transs\,\varthree) \cont \enc\istr \enc{\ntostr\counter} \enc{\wstrl^{\rop}} \enc\wstrr
\\     = \\
        (\la\var \la{\cont'} \la{\istr'} \la{\counter'} \la{\wstrl'} \la{\wstrr'} 
(\la{\counter''}\var\cont' \tuple{\istr',\counter''\csep 
   \wstrl',\cods{\elem'}, \wstrr' \csep \cods{\state'}}) \counter')
   ( \la\varthree\transs\,\varthree) \cont \enc\istr \enc{\ntostr\counter} \enc{\wstrl^{\rop}} \enc\wstrr
\\    \tobdet^{6} \\
  (\la{\counter''}( \la\varthree\transs\,\varthree) \cont  \tuple{\enc\istr,\counter'' \csep \enc{\wstrl^{\rop}} ,\cods{\elem'}, \enc\wstrr \csep \cods{\state'} })\enc{\ntostr\counter}
\\    \tobdet^{2} \\
  \transs\, \cont  \tuple{\enc\istr,\enc{\ntostr\counter} \csep \enc{\wstrl^{\rop}} ,\cods{\elem'}, \enc\wstrr \csep \cods{\state'} }
\\    = \\
  \transs\, \cont  \enc{(\istr,\counter \csep \wstrl , \elem', \wstrr \csep \state')}
\\    = \\
  \transs\, \cont  \enc{\configtwo}
\end{array}$
\end{center}

    \item \emph{The input head does not move and the work head moves left}: if 
      $\delta(\bool,\elem,\state)=(0\csep \elem',\leftarrow \csep \state')$ and $\wstrl = \wstr \elem''$ then:
\begin{center}
$\begin{array}{llllll}
C_{\bool,\elem,\state} ( \la\varthree\transs\,\varthree) \cont \enc\istr \enc{\ntostr\counter} \enc{\wstrl^{\rop}} \enc\wstrr
\\	
=
\\
(\la\var \la{\cont'} \la{\istr'} \la{\counter''} \la{\wstrl'} \la{\wstrr'} 
L \counter') ( \la\varthree\transs\,\varthree) \cont \enc\istr \enc{\ntostr\counter} \enc{\wstrl^{\rop}} \enc\wstrr
\\
    = 
    \\
    (\la\var \la{\cont'} \la{\istr'} \la{\counter'} \la{\wstrl'} \la{\wstrr'} 
 (\la{\counter''}\wstrl' L_0^{\state',\elem'} L_1^{\state',\elem'} L_{\elemblank}^{\state',\elem'}L_{\ems}^{\state',\elem'} \var \cont' \istr' \counter'' \wstrr') \counter')
   ( \la\varthree\transs\,\varthree) \cont \enc\istr \enc{\ntostr\counter} \enc{\wstrl^{\rop}} \enc\wstrr
\\    \tobdet^{6} \\
(\la{\counter''}\enc{\wstrl^{\rop}}  L_0^{\state',\elem'} L_1^{\state',\elem'} L_{\elemblank}^{\state',\elem'}L_{\ems}^{\state',\elem'}  ( \la\varthree\transs\,\varthree) \cont \enc\istr \counter'' \enc\wstrr)\enc{\ntostr\counter}
\\    \tobdet \\
\enc{\wstrl^{\rop}}  L_0^{\state',\elem'} L_1^{\state',\elem'} L_{\elemblank}^{\state',\elem'}L_{\ems}^{\state',\elem'}  ( \la\varthree\transs\,\varthree) \cont \enc\istr \enc{\ntostr\counter} \enc\wstrr
\end{array}$
\end{center}
      
    Two sub-cases, depending on whether $\wstrl$ is an empty or a compound string.
    
    \begin{enumerate}
	 \item \emph{$\wstrl$ is the compound string $\wstr\elem''$}. Then $\wstrl^{\rop} = \elem'' \wstr^{\rop}$ and $\configtwo = (\istr,\counter \csep \wstrl, \elem'',\elem'\wstrr \csep \state')$. The simulation continues as follows:
    \begin{center}
$\begin{array}{llllll}
    = \\
\enc{\elem''\wstr^{\rop}}  L_0^{\state',\elem'} L_1^{\state',\elem'} L_{\elemblank}^{\state',\elem'}L_{\ems}^{\state',\elem'}  ( \la\varthree\transs\,\varthree) \cont \enc\istr \enc{\ntostr\counter} \enc\wstrr

\\    \tobdet^{4} \\
  L_{\elem''}^{\state',\elem'} \enc{\wstr^{\rop}}  ( \la\varthree\transs\,\varthree) \cont \enc\istr \enc{\ntostr\counter} \enc\wstrr
\\    = \\  
        (\la{\wstrl'} \la\var \la{\cont'} \la{\istr'} \la{\counter'}\appendchar{\elem'}(\lambda \wstrr'.\var\cont' 
\tuple{\istr',\counter' \csep \wstrl',\cods{\elem''},\wstrr' \csep \cods{\state'}}))
\enc{\wstr^{\rop}}  ( \la\varthree\transs\,\varthree) \cont \enc\istr \enc{\ntostr\counter} \enc\wstrr
\\    \tobdet^{5} \\
        \appendchar{\elem'}(\lambda \wstrr'. ( \la\varthree\transs\,\varthree) \cont 
\tuple{\enc\istr,\enc{\ntostr\counter} \csep \enc{\wstr^{\rop}},\cods{\elem''},\wstrr' \csep \cods{\state'}})
 \enc\wstrr
 
\\ \reflemmaeq{append-char}     \tobdet^{ O(1) } \\
(\lambda \wstrr'. ( \la\varthree\transs\,\varthree) \cont 
\tuple{\enc\istr,\enc{\ntostr\counter} \csep \enc{\wstr^{\rop}},\cods{\elem''},\wstrr' \csep \cods{\state'}})
 \enc{\elem'\wstrr}
\\     \tobdet^{2} \\
\transs\, \cont 
\tuple{\enc\istr,\enc{\ntostr\counter} \csep \enc{\wstr^{\rop}},\cods{\elem''}, \enc{\elem'\wstrr} \csep \cods{\state'}}
\\     = \\
\transs\, \cont 
\enc{(\istr,\counter \csep \wstr,\elem'', \elem'\wstrr \csep \state')}
\\     = \\
\transs\, \cont \enc{\configtwo}
\end{array}$
\end{center}
    
    	 \item \emph{$\wstrl$ is the empty string $\varepsilon$}. Then $\configtwo = (\istr,\counter \csep \ems, \elemblank,\elem'\wstrr \csep \state')$. The simulation continues as follows:
    \begin{center}
$\begin{array}{llllll}
    = \\
\enc{\ems}  L_0^{\state',\elem'} L_1^{\state',\elem'} L_{\elemblank}^{\state',\elem'}L_{\ems}^{\state',\elem'}   ( \la\varthree\transs\,\varthree) \cont \enc\istr \enc{\ntostr\counter} \enc\wstrr

\\    \tobdet^{4} \\
  L_{\ems}^{\state',\elem'}  ( \la\varthree\transs\,\varthree) \cont \enc\istr \enc{\ntostr\counter} \enc\wstrr
\\    = \\
        ( \la\var \la{\cont'} \la{\istr'} \la{\counter'}
       \appendchar{\elem'}((\la d \lambda \wstrr.\var\cont' 
      \tuple{\istr',\counter'\csep d,\cods{\elemblank},\wstrr \csep  
      \cods{\state'}})\enc\ems))
 ( \la\varthree\transs\,\varthree) \cont \enc\istr \enc{\ntostr\counter} \enc\wstrr
\\    \tobdet^{5} \\
        \appendchar{\elem'}((\la d \lambda \wstrr'. ( 
        \la\varthree\transs\,\varthree) \cont 
\tuple{\enc\istr,\enc{\ntostr\counter} \csep d,\cods{\elemblank},\wstrr' \csep 
\cods{\state'}})\enc\ems)
 \enc\wstrr
 
\\ \reflemmaeq{append-char}     \tobdet^{ O(1) } \\
((\la d \lambda \wstrr'. ( \la\varthree\transs\,\varthree) \cont 
\tuple{\enc\istr,\enc{\ntostr\counter} \csep d,\cods{\elemblank},\wstrr' \csep 
\cods{\state'}})\enc\ems)
 \enc{\elem'\wstrr}
\\     \tobdet^{3} \\
\transs\, \cont 
\tuple{\enc\istr,\enc{\ntostr\counter} \csep \enc\ems,\cods{\elemblank}, 
\enc{\elem'\wstrr} \csep \cods{\state'}}
\\     = \\
\transs\, \cont 
\enc{(\istr,\counter \csep \ems,\elemblank, \elem'\wstrr \csep \state')}
\\     = \\
\transs\, \cont \enc{\configtwo}
\end{array}$
\end{center}
\end{enumerate}
  
    \item \emph{The input head does not move and the work head moves right}:  if 
      $\delta(\bool,\elem,\state)=(0\csep \elem',\rightarrow \csep \state')$ and $\wstrr = \elem'' \wstr $ then:
\begin{center}
$\begin{array}{llllll}
C_{\bool,\elem,\state} ( \la\varthree\transs\,\varthree) \cont \enc\istr \enc{\ntostr\counter} \enc{\wstrl^{\rop}} \enc\wstrr
\\	
     = \\
    (\la\var \la{\cont'} \la{\istr'} \la{\counter'} \la{\wstrl'} \la{\wstrr'} 
 R \counter') 
   ( \la\varthree\transs\,\varthree) \cont \enc\istr \enc{\ntostr\counter} \enc{\wstrl^{\rop}} \enc\wstrr
\\    =\\
    (\la\var \la{\cont'} \la{\istr'} \la{\counter'} \la{\wstrl'} \la{\wstrr'} 
 (\la{\counter''}\wstrr' R_0^{\state',\elem'} R_1^{\state',\elem'} R_{\elemblank}^{\state',\elem'}R_{\ems}^{\state',\elem'} \var \cont' \istr' \counter'' \wstrl') \counter')
   ( \la\varthree\transs\,\varthree) \cont \enc\istr \enc{\ntostr\counter} \enc{\wstrl^{\rop}} \enc\wstrr
\\   \tobdet^{6} \\ 
(\la{\counter''}\enc{\wstrr}  R_0^{\state',\elem'} R_1^{\state',\elem'} R_{\elemblank}^{\state',\elem'}R_{\ems}^{\state',\elem'}  ( \la\varthree\transs\,\varthree) \cont \enc\istr \counter'' \enc{\wstrl^{\rop}})\enc{\ntostr\counter}
\\       \tobdet \\
\enc{\wstrr}  R_0^{\state',\elem'} R_1^{\state',\elem'} R_{\elemblank}^{\state',\elem'}R_{\ems}^{\state',\elem'}  ( \la\varthree\transs\,\varthree) \cont \enc\istr \enc{\ntostr\counter} \enc{\wstrl^{\rop}}
\end{array}$
\end{center}
      
    Two sub-cases, depending on whether $\wstrl$ is an empty or a compound string.
    
    \begin{enumerate}
	 \item \emph{$\wstrr$ is the compound string $\elem''\wstr$}. Then  $\configtwo = (\istr,\counter \csep \wstrl\elem', \elem'',\wstr \csep \state')$. The simulation continues as follows:
    \begin{center}
$\begin{array}{llllll}
   =\\
\enc{\elem''\wstr}  R_0^{\state',\elem'} R_1^{\state',\elem'} R_{\elemblank}^{\state',\elem'}R_{\ems}^{\state',\elem'}  ( \la\varthree\transs\,\varthree) \cont \enc\istr \enc{\ntostr\counter} \enc{\wstrl^{\rop}}

\\    \tobdet^{4} \\
  R_{\elem''}^{\state',\elem'} \enc{\wstr}  ( \la\varthree\transs\,\varthree) \cont \enc\istr \enc{\ntostr\counter} \enc{\wstrl^{\rop}}
\\   =\\  
        (\la{\wstrr'} \la\var \la{\cont'} \la{\istr'} \la{\counter'}\appendchar{\elem'}(\lambda \wstrl'.\var\cont' 
\tuple{\istr',\counter' \csep \wstrl',\cods{\elem''},\wstrr' \csep \cods{\state'}}))
\enc{\wstr}  ( \la\varthree\transs\,\varthree) \cont \enc\istr \enc{\ntostr\counter} \enc{\wstrl^{\rop}}
\\    \tobdet^{5} \\  
        \appendchar{\elem'}(\lambda \wstrl'. ( \la\varthree\transs\,\varthree) \cont 
\tuple{\enc\istr,\enc{\ntostr\counter} \csep \wstrl',\cods{\elem''},\wstr \csep \cods{\state'}})
 \enc{\wstrl^{\rop}}
 
\\ \reflemmaeq{append-char}     \tobdet^{ O(1) } 
(\lambda \wstrl'. ( \la\varthree\transs\,\varthree) \cont 
\tuple{\enc\istr,\enc{\ntostr\counter} \csep \wstrl',\cods{\elem''},\wstr \csep \cods{\state'}})
 \enc{\elem'\wstrl^{\rop}}
\\     \tobdet^{2} \\
\transs\, \cont 
\tuple{\enc\istr,\enc{\ntostr\counter} \csep  \enc{\elem'\wstrl^{\rop}},\cods{\elem''},\wstr \csep \cods{\state'}}
\\    =\\ 
\transs\, \cont 
\tuple{\enc\istr,\enc{\ntostr\counter} \csep  \enc{(\wstrl\elem')^{\rop}},\cods{\elem''},\wstr \csep \cods{\state'}}
\\    =\\ 
\transs\, \cont 
\enc{(\istr,\counter \csep \wstrl\elem',\elem'', \wstr \csep \state')}
\\    =\\ 
\transs\, \cont \enc{\configtwo}
\end{array}$
\end{center}
    
    	 \item \emph{$\wstrr$ is the empty string $\varepsilon$}. Then $\configtwo = (\istr,\counter \csep \wstrl\elem', \elemblank,\ems \csep \state')$. The simulation continues as follows:
    \begin{center}
$\begin{array}{llllll}
   =\\
\enc{\ems}  R_0^{\state',\elem'} R_1^{\state',\elem'} R_{\elemblank}^{\state',\elem'}R_{\ems}^{\state',\elem'}  ( \la\varthree\transs\,\varthree) \cont \enc\istr \enc{\ntostr\counter} \enc{\wstrl^{\rop}}
\\    \tobdet^{4} \\
  R_{\ems}^{\state',\elem'}  ( \la\varthree\transs\,\varthree) \cont \enc\istr \enc{\ntostr\counter} \enc{\wstrl^{\rop}}
\\   =\\  
        (\la\var \la{\cont'} \la{\istr'} \la{\counter'}
       \appendchar{\elem'}((\la d \lambda \wstrl'.\var\cont' 
      \tuple{\istr',\counter'\csep \wstrl',\cods{\elemblank},d \csep  
      \cods{\state'}})\enc\ems)
)
 ( \la\varthree\transs\,\varthree) \cont \enc\istr \enc{\ntostr\counter} \enc{\wstrl^{\rop}}
\\    \tobdet^{5} \\
        \appendchar{\elem'}((\la d \lambda \wstrl'. ( 
        \la\varthree\transs\,\varthree) \cont 
\tuple{\enc\istr,\enc{\ntostr\counter} \csep \wstrl',\cods{\elemblank},d \csep 
\cods{\state'}})\enc\ems)
\enc{\wstrl^{\rop}}
 
\\ \reflemmaeq{append-char}     \tobdet^{ O(1) } \\
((\la d \lambda \wstrl'. ( \la\varthree\transs\,\varthree) \cont 
\tuple{\enc\istr,\enc{\ntostr\counter} \csep \wstrl',\cods{\elemblank},d \csep 
\cods{\state'}})\enc\ems)
 \enc{\elem'\wstrl^{\rop}}
\\    \tobdet^{3} \\ 
\transs\, \cont 
\tuple{\enc\istr,\enc{\ntostr\counter} \csep \enc{\elem'\wstrl^{\rop}},\cods{\elemblank}, \enc{\ems} \csep \cods{\state'}})
\\    =\\ 
\transs\, \cont 
\tuple{\enc\istr,\enc{\ntostr\counter} \csep \enc{(\wstrl\elem')^{\rop}},\cods{\elemblank}, \enc{\ems} \csep \cods{\state'}})
\\    =\\ 
\transs\, \cont 
\enc{(\istr,\counter \csep \wstrl\elem',\elemblank, \ems \csep \state')}
\\    =\\ 
\transs\, \cont \enc{\configtwo}
\end{array}$
\end{center}
\end{enumerate}

    \item \emph{The input head moves left and the work head does not move}: 
    if 
    $\delta(\bool,\elem,\state)=(-1 \csep \elem',\downarrow \csep \state')$,
     then $\configtwo = (\istr,\counter-1 \csep \wstrl, \elem',\wstrr,\state')$. The simulation continues as follows:
\begin{center}
$\begin{array}{llllll}
C_{\bool,\elem,\state} ( \la\varthree\transs\,\varthree) \cont \enc\istr \enc{\ntostr\counter} \enc{\wstrl^{\rop}} \enc\wstrr
\\	
    =\\
    (\la\var \la{\cont'} \la{\istr'} \la{\counter'} \la{\wstrl'} \la{\wstrr'} 
\predl\, S 
   ( \la\varthree\transs\,\varthree) \cont \enc\istr \enc{\ntostr\counter} \enc{\wstrl^{\rop}} \enc\wstrr
   \\ =\\
    (\la\var \la{\cont'} \la{\istr'} \la{\counter'} \la{\wstrl'} \la{\wstrr'} 
\predl\, (\la {\counter''} x\cont' \tuple{\istr',\counter''\csep 
  	\wstrl',\cods{\elem'}, \wstrr' \csep \cods{\state'}} )\counter') 
   ( \la\varthree\transs\,\varthree) \cont \enc\istr \enc{\ntostr\counter} \enc{\wstrl^{\rop}} \enc\wstrr
\\    \tobdet^{6} \\
  \predl\, (\la {\counter''} ( \la\varthree\transs\,\varthree)\cont \tuple{\enc\istr,\counter''\csep 
  	\enc{\wstrl^{\rop}},\cods{\elem'}, \enc\wstrr \csep \cods{\state'}} )\enc{\ntostr\counter}
\\  \reflemmaeq{pred}  \tobdet^{\bigo{\log n}} \\
(\la {\counter''} ( \la\varthree\transs\,\varthree)\cont \tuple{\enc\istr,\counter''\csep 
  	\enc{\wstrl^{\rop}},\cods{\elem'}, \enc\wstrr \csep \cods{\state'}} )\enc{\ntostr{\counter-1}}
	\\   =\\
  \transs\, \cont  \enc{(\istr,\ntostr{\counter-1} \csep \wstrl , \elem', \wstrr \csep \state')}
  \\   =\\
  \transs\, \cont  \enc{(\istr,\ntostr{\counter-1} \csep \wstrl , \elem', \wstrr \csep \state')}
\\   =\\
  \transs\, \cont  \enc{\configtwo}
\end{array}$
\end{center}

    \item \emph{The input head moves left and the work head moves left}: 
    if 
    $\delta(\bool,\elem,\state)=(-1 \csep \elem',\leftarrow \csep \state')$ and $\wstrl = \wstr\elem''$,
     then $\configtwo = (\istr,\counter-1 \csep \wstr, \elem'', \elem'\wstrr,\state')$. The simulation continues as follows:
\begin{center}
$\begin{array}{llllll}
C_{\bool,\elem,\state} ( \la\varthree\transs\,\varthree) \cont \enc\istr \enc{\ntostr\counter} \enc{\wstrl^{\rop}} \enc\wstrr
\\	
    =\\
    (\la\var \la{\cont'} \la{\istr'} \la{\counter'} \la{\wstrl'} \la{\wstrr'} 
\predl\, L \counter') 
   ( \la\varthree\transs\,\varthree) \cont \enc\istr \enc{\ntostr\counter} \enc{\wstrl^{\rop}} \enc\wstrr
\\	
    =\\
    (\la\var \la{\cont'} \la{\istr'} \la{\counter'} \la{\wstrl'} \la{\wstrr'} 
\predl\, (\la{\counter''}\wstrl' L_0^{\state',\elem'} L_1^{\state',\elem'} L_{\elemblank}^{\state',\elem'}L_{\ems}^{\state',\elem'} \var \cont' \istr' \counter'' \wstrr')\counter') 
   ( \la\varthree\transs\,\varthree) \cont \enc\istr \enc{\ntostr\counter} \enc{\wstrl^{\rop}} \enc\wstrr
\\    \tobdet^{6} \\
  \predl\, (\la{\counter''}\enc{\wstrl^{\rop}} L_0^{\state',\elem'} L_1^{\state',\elem'} L_{\elemblank}^{\state',\elem'}L_{\ems}^{\state',\elem'} ( \la\varthree\transs\,\varthree) \cont \enc\istr \counter'' \enc\wstrr)\enc{\ntostr\counter}
  
\\  \reflemmaeq{pred}  \tobdet^{\bigo{\log n}} \\
(\la{\counter''}\enc{\wstrl^{\rop}} L_0^{\state',\elem'} L_1^{\state',\elem'} L_{\elemblank}^{\state',\elem'}L_{\ems}^{\state',\elem'} ( \la\varthree\transs\,\varthree) \cont \enc\istr \counter'' \enc\wstrr)
\enc{\ntostr{\counter-1}}
	\\    \tobdet \\
  \enc{\wstrl^{\rop}} L_0^{\state',\elem'} L_1^{\state',\elem'} L_{\elemblank}^{\state',\elem'}L_{\ems}^{\state',\elem'} ( \la\varthree\transs\,\varthree) \cont \enc\istr \enc{\ntostr{\counter-1}} \enc\wstrr
\end{array}$
\end{center}
And then the case continues with the two sub-cases of case 3 (input head does not move and work head moves left), with the only difference that $\enc{\ntostr{\counter}}$ is replaced by $\enc{\ntostr{\counter-1}}$.

    \item \emph{The input head moves left and the work head moves right}: 
    if 
    $\delta(\bool,\elem,\state)=(-1 \csep \elem',\leftarrow \csep \state')$ and $\wstrl = \wstr\elem''$,
     then $\configtwo = (\istr,\counter-1 \csep \wstr, \elem'', \elem'\wstrr,\state')$. The simulation continues as follows:
\begin{center}
$\begin{array}{llllll}
C_{\bool,\elem,\state} ( \la\varthree\transs\,\varthree) \cont \enc\istr \enc{\ntostr\counter} \enc{\wstrl^{\rop}} \enc\wstrr
\\	
     = 
    \\
    (\la\var \la{\cont'} \la{\istr'} \la{\counter'} \la{\wstrl'} \la{\wstrr'} 
\predl\, R \counter') 
   ( \la\varthree\transs\,\varthree) \cont \enc\istr \enc{\ntostr\counter} \enc{\wstrl^{\rop}} \enc\wstrr
\\	
     = 
     \\
    (\la\var \la{\cont'} \la{\istr'} \la{\counter'} \la{\wstrl'} \la{\wstrr'} 
\predl (\la{\counter''}\wstrr' R_0^{\state',\elem'} R_1^{\state',\elem'} R_{\elemblank}^{\state',\elem'}R_{\ems}^{\state',\elem'} \var \cont' \istr' \counter'' \wstrl'
)\counter') 
   ( \la\varthree\transs\,\varthree) \cont \enc\istr \enc{\ntostr\counter} \enc{\wstrl^{\rop}} \enc\wstrr
\\
    \tobdet^{6} 
\\
  \predl (\la{\counter''}\enc{\wstrr} R_0^{\state',\elem'} R_1^{\state',\elem'} R_{\elemblank}^{\state',\elem'}R_{\ems}^{\state',\elem'} ( \la\varthree\transs\,\varthree) \cont \enc\istr \counter'' \enc{\wstrl^{\rop}})\enc{\ntostr\counter}
  
\\
  \text{by }\reflemmaeq{pred} \ \ \  \tobdet^{\bigo{\log n}}
  \\
(\la{\counter''}\enc{\wstrr} R_0^{\state',\elem'} R_1^{\state',\elem'} R_{\elemblank}^{\state',\elem'}R_{\ems}^{\state',\elem'} ( \la\varthree\transs\,\varthree) \cont \enc\istr \counter'' \enc{\wstrl^{\rop}})
\enc{\ntostr{\counter-1}}
	\\   
	 \tobdet 
	 \\
  \la{\counter''}\enc{\wstrr} R_0^{\state',\elem'} R_1^{\state',\elem'} R_{\elemblank}^{\state',\elem'}R_{\ems}^{\state',\elem'} ( \la\varthree\transs\,\varthree) \cont \enc\istr \enc{\ntostr{\counter-1}} \enc{\wstrl^{\rop}}
\end{array}$
\end{center}
And then the case continues with the two sub-cases of case 4 (input head does not move and work head moves right), with the only difference that $\enc{\ntostr{\counter}}$ is replaced by $\enc{\ntostr{\counter-1}}$.

\item \emph{The input head moves right and the work head does not move}: exactly as case 5 (input head \emph{left}, work head does not move) just replacing $\predl$ with $\succl$ and using \reflemma{succ} instead of \reflemma{pred}.

\item \emph{The input head moves right and the work head moves left}: exactly as case 6 (input head \emph{left}, work head left) just replacing $\predl$ with $\succl$ and using \reflemma{succ} instead of \reflemma{pred}.

\item \emph{The input head moves right and the work head moves right}: exactly as case 7 (input head \emph{right}, work head right) just replacing $\predl$ with $\succl$ and using \reflemma{succ} instead of \reflemma{pred}.\qedhere
\end{enumerate}
 
\end{proof}

Straightforward inductions on the length of executions provide the following 
corollaries.

\begin{corollary}[Executions]
	\label{coro:exec}
	Let $\M$ be a Turing machine. Then there exist a term $\transs$ encoding 
	$\M$ 
	as given by \reflemma{trans-sim} such that
	for every configuration $\config$ of input string $\istr\in\Bool^{+}$
	\begin{enumerate}
		\item \emph{Finite computation}:
		if $\configtwo$ is a final configuration reachable from $\config$ in 
		$n$ 
		transition steps
		then there exists a derivation $\run$ such that $\run:\transs\, \cont 
		\cods\config \tobdet^{\bigo{(n+1)\size\istr\log{\size\istr}}} \cont 
		\cods\configtwo$;
		\item \emph{Diverging computation}: if there is
		no final configuration reachable from $\config$ then $\transs\, \cont 
		\cods\config$ diverges.
	\end{enumerate}
\end{corollary}

\paragraph{The Simulation Theorem.}
We now have all the ingredients for the final theorem of this note.

\begin{theorem}[Simulation]
	Let $f:\Bool^*\rightarrow\Bool$ a function computed by a Turing machine
	$\M$ in time $T_{\M}$. Then there is an encoding $\cods{\cdot}$ into 
	$\detLam$ of $\Bool$, strings, and Turing machines over $\Bool$ such that 
	for every $\istr\in\Bool^+$, there exists $\run$ such that
	$\run:\cods\M \cods\istr \tobdet^n \cods{f(\istr)}$
	where $n=\Theta((T_{\M}(\size\istr)+1)\cdot \size\istr\cdot 
	\log{\size\istr})$.
\end{theorem}

\begin{proof}
	Intuitively, the term is simply
	$$
	\enc\M \defeq \inits (\transs(\finals\Id)
	$$
	where the identity $\Id$ plays the role of the initial continuation.
	
	Such a term however does not belong to the deterministic $\l$-calculus, 
	because the right subterms of applications are not always values. The 
	solution is simple, it is enough to $\eta$-expand the arguments. Thus, 
	define
	
	$$
	\enc\M \defeq \inits (\la\vartwo\transs(\la \var\finals\Id\var)\vartwo)
	$$

	Then
	$$\begin{array}{lll}
		\enc\M \cods{\istr} 
		& = & \\
		\inits (\la\vartwo\transs (\la \var\finals \Id\var)\vartwo) 
		\cods{\istr} 
		& \tobdet^{\Theta(1)} & (\mbox{by }\reflemmaeq{init-config})\\    
		(\la\vartwo\transs (\la \var\finals \Id\var)\vartwo) 
		\cods{\initconfig}
		& \tobdet\\
		\transs (\la \var\finals \Id\var) \cods{\initconfig}
		& \tobdet^{\Theta((T_{\M}(\size\istr)+1)\cdot \size\istr \cdot 
		\log{\size\istr})} & (\mbox{by }\refcoroeq{exec})\\
		(\la \var\finals \Id\var) \cods{\config_{\tt fin}(f(\istr))}
		& \tobdet\\
		\finals \Id \cods{\config_{\tt fin}(f(\istr))}
		& \tobdet^{\Theta(\size\States)} & (\mbox{by 
		}\reflemmaeq{final-config})\\    
		\Id \cods{f(\istr)}
		& \tobdet &\\
		\cods{f(\istr)}
	\end{array}$$
\end{proof}

\bibliographystyle{alpha}
\bibliography{main.bbl}

\end{document}